\documentclass[11pt]{article}
\usepackage[english]{babel}
\usepackage[utf8]{inputenc}
\usepackage[T1]{fontenc}
\usepackage{amsfonts,amsmath,amssymb,amsthm}
\usepackage{enumerate}
\usepackage{caption}
\usepackage{subcaption}
\usepackage{hyperref}
\usepackage{multirow}
\usepackage{manfnt}
\usepackage{graphicx}
\usepackage{dsfont}
\usepackage{float}
\usepackage{enumitem}
\usepackage{a4wide}
\usepackage[ruled,vlined]{algorithm2e}

\newtheorem{theorem}{Theorem}
\newtheorem{claim}[theorem]{Claim}
\newtheorem{lemma}[theorem]{Lemma}
\newtheorem{proposition}[theorem]{Proposition}
\newtheorem{corollary}[theorem]{Corollary}

\def\vect#1{\mathbf{#1}}

\title{The growth rate over trees of any family of set defined by a monadic second order formula is semi-computable}
\author{Matthieu Rosenfeld\footnote{Aix Marseille Université, Université de Toulon, CNRS, LIS, Marseille, France}}
\begin{document}
\maketitle
\begin{abstract}
Monadic second order logic can be used to express many classical notions of sets of vertices of a graph as for instance: dominating sets, induced matchings, perfect codes, independent sets or irredundant sets.
Bounds on the number of sets of any such family of sets are interesting from a combinatorial point of view and have algorithmic applications.  Many such bounds on different families of sets over different classes of graphs are already provided in the literature. 
In particular, Rote recently showed that the number of minimal dominating sets in trees of order $n$ is at most $95^{\frac{n}{13}}$ and that this bound is asymptotically sharp up to a multiplicative constant.
We build on his work to show that what he did for minimal dominating sets can be done for any family of sets definable by a monadic second order formula. 

We first show that, for any monadic second order formula over graphs that characterizes a given kind of subset of its vertices, the maximal number of such sets in a tree can be expressed as the \textit{growth rate of a bilinear system}. This mostly relies on well known links between monadic second order logic over trees and tree automata and basic tree automata manipulations. 
Then we show that this ``growth rate'' of a bilinear system can be approximated from above.
We then use our implementation of this result to provide bounds (some sharp and some almost sharp) on the number of independent dominating sets, total perfect dominating sets, induced matchings, maximal induced matchings, minimal perfect dominating sets,  perfect codes and maximal irredundant sets on trees.
We also solve a question from D. Y. Kang et al. regarding $r$-matchings and improve a bound from Górska and Skupień on the number of maximal matchings on trees.
Remark that this approach is easily generalizable to graphs of bounded tree width or clique width (or any similar class of graphs where tree automata are meaningful).
\end{abstract}

\section{Introduction}
Monadic second order logic can be used to express many classical notions of sets of vertices of a graph as for instance: dominating sets, induced matchings, perfect codes, independent sets or irredundant sets.
Bounds on the number of such subsets are interesting from a combinatorial point of view and have algorithmic applications. 
Lower bounds on the number of such sets have direct implications on the enumeration complexity, but the range of algorithmic applications is much wider than that.
For instance, the celebrated upper-bound by Moon and Moser of $3^{\frac{n}{3}}$
on the number of maximal independent sets in a graph of order $n$ 
was used by Lawler to give a graph coloring algorithm in 
time $O^*((1+3^{\frac{1}{3}})^n)$ which was the fastest coloring algorithm for 25 years \cite{Moon1965,Lawler1976ANO}.
Eppstein improved the algorithm running time in 2003 by using an upper-bound on the number of small maximal independant sets \cite{eppstein} (this result as since been improved a few times).

One easily verifies that the result of Moon and Moser is sharp since a collection of triangle has $3^{\frac{n}{3}}$ maximal independent sets. The same question on other subsets (in particular variations of dominating sets) received a lot of attention, but we do not have many sharp bounds. For instance, it was showed in \cite{mindomsetgraphs} that there are at most $1.7159^n$ minimal dominating sets in a graph of order $n$, but no graph is known with more that $1.5704^n$ minimal dominating sets. On the other hand, it is easier to obtain sharp bounds when restricting to some smaller classes of graphs \cite{COUTURIER2015634, indepDom}. Trees is a natural class of graphs for this kind of questions.
Recently Rote showed that the maximal number of minimal dominating sets of trees of order $n$ grows in  $\theta(95^{\frac{n}{13}})$ \cite{r-mnmds-19,arxivRote}.
Given a family of graphs $\mathcal{G}$ and a function $\mathcal{F}$ that maps every graph to a family of subsets of its vertices, we call the growth rate of $\mathcal{F}$ over $\mathbb{G}$ the quantity given by $$\limsup_{n\rightarrow\infty}\max_{G\in\mathcal{G},|G|=n} |\mathcal{F}(G)|.$$
The result from Rote implies that the growth rate of the number of minimal dominating sets over trees is $95^{\frac{1}{13}}$.
On the other hand, Golovach et Al. computed the growth rate of some families of sets over paths \cite{sigmarhopath}.
They used an automatic approach to compute the growth rate of different (maximal or minimal) $(\sigma,\rho)$-dominating sets over paths. They also found a recurrence relation and computed the growth rate of maximal irredundant sets over paths.
Here, we generalize and automatize the approach of Rote  and we show that the growth rate over trees of any family of set that can be defined in monadic second order logic is semi-computable. That is, there exists an algorithm that takes as input the monadic second order formula describing the formula and that outputs a decreasing sequence of upper-bounds of the growth rate that converges toward the actual value of the growth rate. For some families, we are able to use the algorithm to precisely compute the growth rate.

We start by introducing some definitions and notations in Section \ref{secdef}.
Then, we show in Section \ref{MSOandtreeautomaton} that given an MSO formula over graphs with a free second order variable there is a tree automaton that recognizes exactly the pairs $(T,S)$ such that $T$ is a tree and $S$ is a subset of $T$ that satisfies the formula.
In Section \ref{msotolinearalg}, we show that the number of terms accepted by the tree automaton has a bilinear inductive expression. This bilinear inductive expression can be easily computed from the tree automaton.
Up to this point everything can be considered as folklore, but in the rest of Section \ref{msotolinearalg} we show that the growth rate of this quantity is semi-computable.
The computation of this quantity relies on finding convex polytopes that are fixed point of the bilinear map. This idea was used by Rote  \cite{r-mnmds-19,arxivRote} and we show that in general it can  be used to approximate the growth rate of a bilinear system from above.
Moreover, examination of the vertices of the polytope provide extrem examples of trees and it allowed us to find constructions that are out of reach of a simple exhaustive search.

In Section \ref{results}, we use an implementation of the decision procedure to provide some bounds. In particular, we compute sharp bounds for the number of independent dominating sets, total perfect dominating sets, $r$-matching for some values of $r$, minimal perfect dominating sets and perfect codes. 
Rote asked whether one could use his approach to compute a sharp upper bound on the number of maximal irredundant sets \cite{r-mnmds-19,arxivRote}. We were not able to do so but we show that the growth rate of the maximal number of maximal irredundant sets is between $\frac{14}{9}\approx1.555556$ and $48^{\frac{1}{9}}\approx 1.53746$. We also answer a question from D. Y. Kang et al. asking for which $r\in\{3,4,5,7,9\}$ is the maximal number of $r$-matchings over trees reached by paths \cite{inducedmatching}. The authors of \cite{GORSKA20071367} gave bounds on the maximal number of maximal matchings of trees and we tighten the bounds that they provided. We also obtain good bounds for the growth rate of the number of maximal induced matchings in trees (the gap between the upper and the lower bound is less than $10^{-25}$).  

We end our paper with a discussion regarding the possible generalizations. 
In particular, we mention the fact that the approach is easily generalizable to graphs of tree-width (or clique-width) at most $k$ for any constant $k$ and that monadic second order logic can be replaced by counting monadic second order logic.

\section{Definitions and notations}\label{secdef}
For $X\in\{\mathbb{Z},\mathbb{Q},\mathbb{R}\}$, we denote by $X_{\ge0}$ (resp. $X_{>0}$) the subset of all the non-negative (resp. positive) elements of $X$.
The \emph{order} of a graph $G$ denoted by $|G|$ is the number of vertices of $G$.

\subsection{Monadic second order Logic, definable sets and growth rate}

Monadic second-order logic is a restriction of second-order logic where the second-order quantification is restricted to quantification over sets.
We restrict ourselves to $MSO_1$, that is second order quantifiers can only be used over sets of vertices but not over sets of edges. The syntax of $MSO_1$ is given by
$$\phi:= E(x,y) | x=y | X=Y | x\in X | \lnot \phi | \phi_1 \land \phi_2|  \phi_1 \lor \phi_2 | \phi_1 \Rightarrow \phi_2 
|\forall x.\phi|\forall X.\phi | \exists x.\phi |\exists X.\phi,$$ where lower-case letters are first order variables and upper-case letters are second order variables (with the exception of $E$ that denotes the adjacency relation). 
The semantic is defined as expected by interpreting $E(x,y)$ as the adjacency relation. 

A graph $G$ \emph{models} a closed formula $\Psi$, if $\Psi$ is true when interpreted over $G$ which we denote by $G\models \Psi$. Similarly, given a formula $\Psi$ with one free variable $X$, a graph $G$ and a set $S$ of vertices of $G$, we say that $G,S$ \emph{model} $\Psi$ (and we write $G,S\models \Psi$) if $\Psi$ is true when interpreted over $G$ with $X$ interpreted as $S$.

Let $D$ be a function that maps any graph to a family of subsets of its vertices (we will call such a function a \emph{family of sets}).
We say that $D$ is \emph{$MSO_1$ definable} if there exists an $MSO_1$ formula $\Psi$ with one free variable of the second order,
such that for all graph $G$, $$D(G)=\{S\subseteq G: G,S\models \Psi\}.$$ 
For instance, a set $S$ is a dominating set of a graph if every vertex is in $S$ or has a neighbor in $S$. Hence dominating sets can be defined by the $MSO$ formula $$\forall x, (x\in S \lor \exists y, (y\in S \land E(x,y)))\,.$$

The \emph{growth rate} of $D$ (or of $\Psi$) over a family of graphs $\mathcal{G}$ is the quantity defined as:
$$\gamma(D)=\limsup_{l\rightarrow \infty}\max_{G\in \mathcal{G}, |G|=l} \left| D(G)\right|^\frac{1}{l}$$
or equivalently
$$\gamma(\Psi)=\limsup_{l\rightarrow \infty}\max_{G\in \mathcal{G}, |G|=l} |\left\{ S\subseteq G:G,S\models \Psi\right\}|^\frac{1}{l}\,.$$
We will show that over trees this quantity is approximable from above. 

\section{Reduction to WS2S and tree automata}\label{MSOandtreeautomaton}
In this section, we use the decidability of WS2S to show that the family of sets that satisfy any given $MSO_1$ formula is recognized by a deterministic binary-tree automaton.
In particular, only Corollary \ref{corphitoA} and the definition of deterministic binary-tree automaton are useful for the other sections.
Moreover, the content of this whole section is folklore and can be retrieved from \cite{tata} so we allow ourselves not to be absolutely formal.

First, we need to introduce WS2S (Weak Monadic second order logic with 2 successors).\footnote{The definition we give here is not exactly WS2S.
Formally, the terms of WS2S are words over $\{0,1\}$, each such words is to be understood as a position in the infinite binary-tree ($01$ would be the left right child of the left child of the root). By using a slightly different convention we avoid the task of defining a binary tree as being a set of words. For more details on WS2S see \cite{tata} for instance.}
The objects of WS2S are \emph{binary ordered trees}, i.e., rooted trees in which each node has at most one right child and at most one left child.
The first order elements are the vertices of the tree and the second order elements are sets of vertices.
The syntax of formulas of WS2S is given by:
$$\phi:=  x=y | X=Y | x\in X |x\operatorname{ch}_1y|x\operatorname{ch}_2y|\lnot \phi | \phi_1 \land \phi_2|  \phi_1 \lor \phi_2 | \phi_1 \Rightarrow \phi_2 
|\forall x.\phi|\forall X.\phi | \exists x.\phi |\exists X.\phi$$ where $\phi_1$
where $x$ and $y$ are first order variables and $X$ and $Y$ are second order variables and $\phi_1$ and $\phi_2$ are formulas
The semantic is defined as expected by interpreting $\operatorname{ch}_1$ and $\operatorname{ch}_2$ to be respectively the left and right child relations (i.e., $x\operatorname{ch}_1y$ is true if $y$ is the right left child of $x$). If a closed formula $\Psi$ is true when interpreted over a binary ordered trees $T$, we say that $T$ models $\Psi$ and we write $T\models\Psi$.

This notion will be useful as we will see later any set defined by a WS2S formula is recognizable by a deterministic binary-tree automaton. Let us first show that we can reduce any $MSO_1$ formula to an equivalent WS2S formula.
\begin{figure}
\centering
    \includegraphics{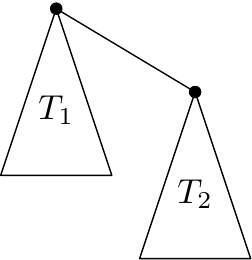}
    \caption{Illustration of J\label{illJ}}
\end{figure}
\begin{figure}
\centering
\begin{minipage}[c]{0.45\textwidth}
\centering
 \includegraphics[scale=1.5,page=2]{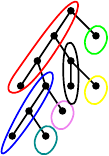}
 \subcaption{The proper binary ordered tree}
    
\end{minipage}
\begin{minipage}[c]{0.45\textwidth}
\centering
 \includegraphics[scale=3,page=11]{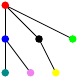}
 \subcaption{The rooted tree}
\end{minipage}
\caption{Illustration of the surjection from proper binary ordered trees to rooted trees\label{illbij}}
\end{figure}

 \newcommand\bijJ{\vartheta}
 \newcommand\basetree{\perp}

Let $J$ be the operation, illustrated in Figure \ref{illJ}, from pairs of rooted trees to rooted trees such that for all trees $T_1,T_2$, $J(T_1,T_2)$ is obtained by taking the disjoint union of $T_1$ and $T_2$ adding and edge between the roots of $T_1$ and $T_2$ and where the root of $J(T_1,T_2)$ is the root of $T_1$.
It is clear that every rooted tree can be obtained by applying $J$ multiple times starting from copies of the singleton tree. In other words, if we let $\basetree$ be the singleton tree, then the set of all rooted trees is the smallest set that contains $\basetree$ and that is closed under $J$.
Thus there is a natural surjections from terms over $\{J,\basetree\}$ to rooted trees. Moreover, the syntax tree of terms over $\{J,\basetree\}$ gives a simple bijection between terms over  $\{J,\basetree\}$ and proper binary ordered trees (a rooted binary tree is \emph{proper} if each node has either 0 or 2 children).

This gives us a natural surjection $\bijJ$ from proper binary  ordered trees to trees (Figure \ref{illbij} is an example of this bijection). 
Given a proper binary ordered tree $T$,  we also have a natural bijection 
$\bijJ_T$ that maps the leaves of $T$ to the nodes of $\bijJ(T)$. 

Given an ordered binary tree $T$ and two nodes $t_1$ and $t_2$, we say that the node $t_1$ is a \emph{right ancestor} of 
$t_2$ if $t_1=t_2$ or if the left child of $t_1$ is a right ancestor of $t_2$. The following Claim, illustrated by the colors in Fig. \ref{illbij}, is a simple consequence of the definition.
\begin{claim}\label{rightancestors}
Let $T$ be a proper binary ordered tree and  $u,v$ be two leaves of $T$. Then $\bijJ_T(u)$ and $\bijJ_T(v)$ are adjacent in $\bijJ(T)$ 
if and only if there exists $u'$ and $v'$ in $T$ that are respectively right-ancestors of $u$ and $v$ such that one of $u'$ and $v'$ is the right child of the other one in $T$.
\end{claim}
The bijection $\bijJ$ and this Claim are the main tools that allow us to translate an $MSO_1$ formula to a $WS2S$ formula in Theorem \ref{MSO_1toWS2S}. 

We will say that a set $S$ of a proper binary ordered tree is \emph{consistent} if every vertex that is not a leaf is not in $S$.
A consistent set $S$ of a proper binary ordered tree $T$ naturally induces a subset $\bijJ_T(S)$ of $\bijJ(T)$.

\begin{theorem}\label{MSO_1toWS2S}
For any $MSO_1$ formula $\Psi$ on graphs there exists a WS2S formula $\Psi_2$ that can be computed from $\Psi$ such that for all proper binary ordered tree $T$ and every subset $S$ of $T$ the following are equivalent:
\begin{itemize}
 \item $S$ is consistent and $\bijJ(T),\bijJ_T(S) \models \Psi$,
 \item $T,S \models \Psi_2$.
\end{itemize}
\end{theorem}
\begin{proof}
Let $C:= \forall x,y: x\operatorname{ch}_1y \implies x\not\in S $. This formula is satisfied if and only if $S$ is consistent. Thus $\Psi_2$ will be the conjunction of $C$ and a formula that is true if and only if  $\bijJ(T),\bijJ_T(S) \models \Psi$ under the assumption that $S$ is consistent. Let us now build the second part of this formula.
 
Let $$LCC(X):=\forall x,y: (x\in X\land x\operatorname{ch}_1y)\implies y\in X$$ and 
$$RA(x,y):= \forall A:(LCC(A)\land y\in A)\implies x\in A.$$
Then $LCC(X)$ is true if the set $X$ is closed by taking the left child and $RA(x,y)$ is true if any set that contains $y$ and is closed by taking the left child also contains $x$. 
That is $RA(x,y)$ is true if and only if $y$ is a right ancestor of $x$.
Now let
$$Edge(u,v):= \exists u',v': RA(u,u')\land RA(v,v')\land (u'\operatorname{ch}_1v'\lor v'\operatorname{ch}_1u').$$
From Claim \ref{rightancestors}, we know that for any leaves $u,v\in T$,  we have $Edge(u,v)$ is true if and only if $\bijJ_T(u)$ and $\bijJ_T(v)$ share an edge in $\bijJ(T)$. In other words, $E(\bijJ_T(u),\bijJ_T(v))$ is true if and only if $Edge(u,v)$ is true.

Let $$Leaf(x):=\forall y: \lnot ( x\operatorname{ch}_1y) \land \lnot (x\operatorname{ch}_2y)$$ and $$Leaf(X):=\forall x: x\in X\implies leaf(x)$$
be the two formulas satisfied if the variable is a leaf (resp. a set of leaves) of $T$.

Let $\Psi_2$ be the conjunction of $C$ and $\Psi$ where each quantification is replaced by a quantification restricted to leaves and by replacing every use of the relation $E(x,y)$ by the formula $Edge(x,y)$.
More formally, let $g$ be the function from $MSO_1$ formulas over graphs to $WS2S$ formulas inductively defined by:
\begin{itemize}
 \item $g(E(x,y))=Edge(x,y)$ and $g(t)=t$ for any other atomic formula,
 \item for any first or second order variable $\mu$, then $g(\forall \mu.\phi)=\forall \mu.(Leaf(\mu)\implies g(\phi)),$
 \item for any first or second order variable $\mu$, then $g(\exists \mu.\phi)=\exists \mu.(Leaf(x)\land g(\phi)),$
 \item $g(\lnot(\phi))=\lnot(g(\phi))$,
\item  $g(\phi_1 \star \phi_2)= g(\phi_1)\star g(\phi_2)$ for any binary logical connector $\star\in\{\land,\lor,\implies\}$.
\end{itemize}
Then $\Psi_2=C\land g(\Psi)$ and it easily verified by induction on $g$ that $\Psi_2$ has the desired property.
\end{proof}

A \emph{deterministic binary-tree automaton (\textbf{DTFA})} is a tuple $\mathcal{A}=(Q, \Sigma, Q_f, \delta)$, where $Q$ is a set of states, $\Sigma=\Sigma_0\cup\Sigma_2$ is an alphabet with letters of arity $0$ and $2$, $Q_f\in Q$ is a set of final states, and $\delta: \Sigma_0\cup(\Sigma_2\times Q\times Q)\rightarrow Q$ is a transition function. We let $\hat\delta $ be the function such that for any $\alpha \in \Sigma_2$ and terms $x_1,x_2$ over $\Sigma$, $\hat\delta(\alpha(x_1,x_2))=\delta(\alpha,\hat\delta(x_1),\hat\delta(x_2))$, and $\hat\delta(\beta)=\delta(\beta)$ for any $\beta\in\Sigma_0$. A term $t$ is then accepted by $\mathcal{A}$ if $\hat\delta(t)\in Q_f$.
Intuitively the automaton works on the syntax tree of the term, the state of a node is inductively deduced from the states of its children and from its symbol, and the term is accepted if the state of the root is in $Q_f$.

As we have seen before the terms over $\{J,\basetree\}$ are in natural bijection with their syntax trees and it we say that a proper binary ordered tree is accepted by a DTFA if and only if the associated term is accepted by the DTFA. We need to extend the definition to pairs $(T,S)$ where $T$ is a proper binary ordered tree and $S$ is a subset of the vertices of $T$. We let $t(T,S)$ be the term over $\{J_0,J_1,\basetree_0, \basetree_1\}$ (where $J_0$ and $J_1$ are of arity $2$ and $\basetree_0$ and $\basetree_1$ are of arity $0$) such that $T$ is isomorphic to the syntax tree of $t(T,S)$ and the index of a letter is $1$ if and only if the corresponding node of $T$ belongs to $S$.
\begin{theorem}\label{WS2Stotrees}
 For any WS2S formula $\Psi$ with a free second-order variable $X$, there exists a DTFA $\mathcal{A}$ over $\{J_0,J_1,\basetree_0, \basetree_1\}$ such that for any proper binary ordered tree $T$ and any subset $S\in T$ the following are equivalent:
 \begin{itemize}
  \item $T,S\models \Psi$ 
  \item $t(T,S)$ is accepted by  $\mathcal{A}$.
 \end{itemize}
 Moreover, $\mathcal{A}$ can be computed from $\Psi$.
\end{theorem}

We do not provide the details of the proof of this result, but it can be deduced by assembling different results from \cite{tata} (this result is almost a particular case of Lemma 3.3.4).
Remark, that the natural statement would replace proper binary ordered tree by binary ordered tree, but since being proper for a binary tree is expressible in MSO this is not a problem. Combining Theorem \ref{MSO_1toWS2S} and Theorem \ref{WS2Stotrees} gives the following result.
\begin{corollary}\label{corphitoA}
 For any $MSO_1$ formula $\Psi$ with a second-order free variable $X$ there exists a DTFA $\mathcal{A}_\Psi$ over $\{J_0,J_1,\basetree_0, \basetree_1\}$ that can be computed from $\Psi$ such that for any proper binary ordered tree $T$ and any subset $S\in T$ the following are equivalent:
 \begin{itemize}
  \item $S$ is consistent and $\bijJ(T), \bijJ_T(S)\models \Psi$,
  \item $t(T,S)$ is accepted by  $\mathcal{A}_\Psi$.
 \end{itemize}
\end{corollary}

\section{The maximal number of sets}\label{msotolinearalg}
\subsection{An explicit formula...}
Let $\mathcal{T}$ (resp. $\mathcal{T}_n$) be the set of trees (resp. of trees of order $n$).
Let $\mathcal{B}$ (resp. $\mathcal{B}_n$) be the set of  proper binary ordered tree (resp. that contains exactly $n$ leaves).
For any $T\in \mathcal{B}$ and any DTFA $\mathcal{A}$, let
$Accepted_{\mathcal{A}}(T)$ be the set of sets $S$ of vertices of $T$ such that $t(T,S)$ is accepted by $\mathcal{A}$.

\begin{lemma}\label{psitoA}
 For any $MSO_1$ formula $\Psi$ with a second-order free variable $X$ there exists a DTFA $\mathcal{A}_\Psi$ that can be computed from $\Psi$ such that for any integer $n$:
 $$\max_{T\in\mathcal{T}_n} \left|\{S\subseteq T: T,S\models\Psi\}\right|= \max_{T\in\mathcal{B}_n}\left|Accepted_{\mathcal{A}_\Psi}(T)\right|.$$
\end{lemma}
\begin{proof}
Recall that $\bijJ$ is a surjection from $\mathcal{B}$ to $\mathcal{T}$. 
Since a binary ordered tree with $n$ leaves is mapped to a tree of order $n$ by $\bijJ$, 
we deduce that the restriction of $\bijJ$ to $\mathcal{B}_n$ is a surjection from $\mathcal{B}_n$ to $\mathcal{T}_n$.
Thus $\mathcal{T}_n =\{\bijJ(T):T\in\mathcal{B}_n\}$.
Moreover, for any $T\in \mathcal{B}_n$, $\bijJ_T$ is a bijection from the set of consistent subset of $T$ to the sets of $\bijJ(T)$. Thus we get the following equality 
 $$\max_{T\in\mathcal{T}_n} \left|\{S\subseteq T: T,S\models\Psi\}\right|= \max_{T\in\mathcal{B}_n}\left|\{S\subseteq T: S \text{ is consistent and }\bijJ(T), \bijJ_T(S)\models \Psi\}\right|.$$
We now get our result by letting $\mathcal{A}_\Psi$ be the automaton given by Corollary \ref{corphitoA} . 
\end{proof}

For any proper binary ordered tree $T$ and any DTFA $\mathcal{A}=(Q, \Sigma, Q_f, \delta)$, we let $T^\mathcal{A}\in \mathbb{R}_{\ge 0}^{|Q|}$ be the vector such that for any $q\in Q$,
$(T^\mathcal{A})_q$ is the number of subsets of $T$ such that $\mathcal{A}$ ends-up in the state $q$, that is: 
$$(T^\mathcal{A})_q= |\{S\subseteq T: \hat\delta(t(T,S))=q\}|.$$

If moreover, we let $\mathbf{F}^\mathcal{A}\in \mathbb{R}_{\ge 0}^{|Q|}$ be the indicator vector of $Q_f$ (i.e., $(\mathbf{F}^\mathcal{A})_q=1$ if $q\in Q$ and $(\mathbf{F}^\mathcal{A})_q=0$ otherwise), then:
\begin{equation}
|Accepted_{\mathcal{A}}(T)|=\mathbf{F}^\mathcal{A}\cdot T^\mathcal{A}.\label{nbsolsfromtreevec}
\end{equation}
For any DTFA $\mathcal{A}=(Q, \Sigma, Q_f, \delta)$, we let $\widetilde{\mathcal{A}}:\mathbb{R}_{\ge 0}^{|Q|}\times\mathbb{R}_{\ge 0}^{|Q|}\rightarrow\mathbb{R}_{\ge 0}^{|Q|}$ be the bilinear map such that for all $\mathbf{u},\mathbf{v}\in\mathbb{R}_{\ge 0}^{|Q|}$ and $q\in Q$:
$$\widetilde{\mathcal{A}}(\mathbf{u},\mathbf{v})_q=\sum_{\delta(J_0,q_1,q_2)=q}\mathbf{u}_{q_1}\cdot\mathbf{v}_{q_2}$$

Then by construction for any proper binary ordered tree $T$ whose left subtree is $T_1$ and whose right subtree is $T_2$, 
$$T^\mathcal{A}=\widetilde{\mathcal{A}}\left(T_1^\mathcal{A},T_2^\mathcal{A}\right)$$
Remark that only the leaves of $T$ can be in $S$ (since we want $S$ to be consistent), and thus we know that there is no transition with $J_1$ in our tree-automaton and this is why only $J_0$ appears in the definition of $\widetilde{\mathcal{A}}$.

Finaly, let $I^\mathcal{A}\in\mathbb{R}_{\ge 0}^{|Q|}$ be the vector such that $I^\mathcal{A}=(T_0)^A$ where $T_0$ is the tree that contains only a root with no child or in other words, for all $q\in Q$
$$(I^\mathcal{A})_q=|\{x\in \{0,1\}: \delta(x)=q\}|.$$

For any  proper binary ordered tree $T$, there is a natural associated term $t_T$ over $\{\widetilde{\mathcal{A}},I^\mathcal{A}\}$ (we replace the leaves by $I^\mathcal{A}$ and the internal nodes by $\widetilde{\mathcal{A}}$). If $\hat{t_T}$ is the result of the evaluation of $t_T$ seen as an expression then it is easy to show by induction that equation \eqref{nbsolsfromtreevec} can be rewritten
\begin{equation}
 |Accepted_{\mathcal{A}}(T)|=\mathbf{F}^\mathcal{A}\cdot \hat{t_T}.\label{nbsolsfromterm}
\end{equation}

Since the number of internal nodes of any tree $T\in\mathcal{B}_n$ is $n-1$, the number of occurrences of $\widetilde{\mathcal{A}}$ in $t_T$ is $n-1$.

Given any bilinear map $\mathbf{B}:\mathbb{R}_{\ge 0}^{n}\times\mathbb{R}_{\ge 0}^{n}\rightarrow\mathbb{R}_{\ge 0}^{n}$ and any vector $\mathbf{V_0}\in\mathbb{R}_{\ge 0}^{n}$, we let $\mathbf{B}^k(\mathbf{V_0})$ be the set of vectors obtained by any expression over $\mathbf{B}$ and $\mathbf{V_0}$ with $k-1$ occurrences of $\mathbf{B}$ (or equivalently with $k$ occurrences of $\mathbf{V_0}$). In other words, the $\mathbf{B}^k(\mathbf{V_0})$ are inductively defined by:
  \begin{itemize}
    \item $\mathbf{B}^1(\mathbf{V_0})=\{\mathbf{V_0}\}$,
    \item for all $k$, $\mathbf{B}^k(\mathbf{V_0})=\bigcup\limits_{i=1}^{k-1}\left\{\mathbf{B}(x,y): x\in\mathbf{B}^i(\mathbf{V_0}), y\in\mathbf{B}^{k-i}(\mathbf{V_0})\right\}$.
  \end{itemize}

Then by construction and by equation \eqref{nbsolsfromterm}
$$\left\{|Accepted_{\mathcal{A}}(T)|: T\in B_n\right\}=\left\{\mathbf{F}^\mathcal{A}\cdot\mathbf{v}:\mathbf{v}\in\widetilde{\mathcal{A}}^k(I^\mathcal{A})\right\}$$
The following lemma is a direct consequence of this last equation and of Lemma \ref{psitoA}.
\begin{lemma}
 For any $MSO_1$ formula $\Psi$ with a second-order free variable $X$ there exists an integer $n$, a bilinear map $\mathbf{B}:\mathbb{R}_{\ge 0}^{n}\times\mathbb{R}_{\ge 0}^{n}\rightarrow\mathbb{R}_{\ge 0}^{n}$ and  vectors $\mathbf{V_0},\mathbf{F}\in\mathbb{R}_{\ge 0}^{n}$ that can be computed from $\Psi$ such that for any integer $k$:
 $$\max_{T\in\mathcal{T}_k} \left|\{S\subseteq T: T,S\models\Psi\}\right|= \max_{\mathbf{v}\in\mathbf{B}^k(\mathbf{V_0})}\left|\mathbf{F}\cdot\mathbf{v}\right|.$$
\end{lemma}

It implies the following corollary.
\begin{corollary}\label{formulatolinearalg}
 For any $MSO_1$ formula $\Psi$ with a second-order free variable $X$ there exist an integer $n$, a bilinear map $\mathbf{B}:\mathbb{R}_{\ge 0}^{n}\times\mathbb{R}_{\ge 0}^{n}\rightarrow\mathbb{R}_{\ge 0}^{n}$ and  vectors $\mathbf{V_0},\mathbf{F}\in\mathbb{R}_{\ge 0}^{n}$ that can be computed from $\Psi$ such that:
 $$\limsup\limits_{k\rightarrow \infty}\max_{T\in\mathcal{T}_k} \left|\{S\subseteq T: T,S\models\Psi\}\right|^{\frac{1}{k}}= \limsup\limits_{k\rightarrow \infty}\max_{\mathbf{v}\in\mathbf{B}^k(\mathbf{V_0})}\left|\mathbf{F}\cdot\mathbf{v}\right|^{\frac{1}{k}}.$$
\end{corollary}
\subsection{...expressed as ``the growth rate'' of a bilinear system}
For any integer $n$, bilinear map $\mathbf{B}:\mathbb{R}^{n}\times\mathbb{R}^{n}\rightarrow\mathbb{R}^{n}$\footnote{We only need to compute the growth rate of bilinear systems with non-negative coefficients which simplifies some parts. However, whenever we the non-negativity condition does not help we try to stay as general as possible.} and vectors $\mathbf{V_0},\mathbf{F}\in\mathbb{R}^{n}$, we call the triple $(\mathbf{B},\mathbf{V_0},\mathbf{F})$ \emph{a bilinear system}. The growth rate of the bilinear system
$(\mathbf{B},\mathbf{V_0},\mathbf{F})$ denoted by $\rho(\mathbf{B},\mathbf{V_0},\mathbf{F})$ is defined as
$$ \rho(\mathbf{B},\mathbf{V_0},\mathbf{F})=\limsup\limits_{k\rightarrow \infty}\max_{\mathbf{v}\in\mathbf{B}^k(\mathbf{V_0})}\left|\mathbf{F}\cdot\mathbf{v}\right|^{\frac{1}{k}}.$$

From Corollary \ref{formulatolinearalg}, we want to compute growth rates of bilinear systems and this subsection is devoted to results regarding the computability of these growth rates.
In particular, we will show that we can approximate this quantity from above and we will give a criterion that can provide exact value under some conditions. 
Let us start with the following Lemma that can be easily deduced from the bilinearity of $\mathbf{B}$.
\begin{lemma}\label{simpleBilinearity}
Let $n$ be an integer, $\mathbf{B}:\mathbb{R}^{n}\times\mathbb{R}^{n}\rightarrow\mathbb{R}^{n}$ be a  bilinear map and $\mathbf{V_0},\mathbf{F}\in\mathbb{R}^{n}$ be vectors.
For any positive constant $\alpha\in\mathbb{R}_{>0}$, then:
$$\rho\left(\frac{\mathbf{B}}{\alpha},\mathbf{V_0},\mathbf{F}\right)=\rho\left(\mathbf{B},\frac{\mathbf{V_0}}{\alpha},\mathbf{F}\right)=\frac{\rho(\mathbf{B},\mathbf{V_0},\mathbf{F})}{\alpha}$$.
\end{lemma}
This lemma tells us that if we can semi-decide whether the growth rate of a bilinear system is greater (resp. smaller) than $1$, we can  semi-decide whether the growth rate of a bilinear system is greater (resp. smaller) than any fixed constant.

The next step is to show that without lose of generality we can assume some nice properties regarding the bilinear map.
Let $n$ be an integer, $\mathbf{B}:\mathbb{R}^{n}\times\mathbb{R}^{n}\rightarrow\mathbb{R}^{n}$ be a  bilinear map and $\mathbf{V_0},\mathbf{F}\in\mathbb{R}^{n}$ be vectors.
We say that $i$-th coordinate is \emph{accessible} if there is an integer $k$ and $\mathbf{v}\in\mathbf{B}^k(\mathbf{V_0})$ such that the $i$-th coordinate of $\vect{v}$ is non zero.
Let $\vect{e_i}$ be the vector whose $i$th coordinate is $1$ and the others are $0$. 
We define inductively the set of \emph{co-accessible} coordinates:
\begin{itemize}
 \item if the $i$-th coordinate of $\vect{F}$ is non zero then $i$ is co-accessible,
 \item if one of the co-accessible coordinate of $\mathbf{B}(\vect{e_i}, \vect{e_j})$ is non zero and the $j$-th coordinate is accessible then  the $i$-th coordinate is co-accessible,
 \item if one of the co-accessible coordinate of $\mathbf{B}(\vect{e_i}, \vect{e_j})$ is non zero and the $i$-th coordinate is accessible then the $j$-th coordinate is co-accessible.
\end{itemize}
The motivation of these definitions is that we can ignore coordinates that are not accessible and co-accessible since they do not influence the result.
Let $\widetilde{m}$ be the number of accessible and co-accessible coordinates.
Let $h$ be the endomorphism that maps the $i$-th coordinate to the $i$-th  accessible and co-accessible coordinate and $h^T$ be the endomorphism that maps the $i$-th accessible and co-accessible coordinate to the $i$-th coordinate. 
Let $\widetilde{\mathbf{B}}=h\circ\mathbf{B}\circ(h^T\times h^T)$, $\widetilde{\vect{v}}=h(\vect{v})$ and $\widetilde{\vect{F}}=h(\vect{F})$.
The following lemma is a direct consequence of this definition.
\begin{lemma}
 Let $n$ be an integer, $\mathbf{B}:\mathbb{R}^{n}\times\mathbb{R}^{n}\rightarrow\mathbb{R}^{n}$ be a  bilinear map and $\mathbf{V_0},\mathbf{F}\in\mathbb{R}^{n}$ be vectors.
 Then for any $k$, 
 $$\max_{\mathbf{v}\in\mathbf{B}^k(\mathbf{V_0})}\left|\mathbf{F}\cdot\mathbf{v}\right|=
 \max_{\mathbf{v}\in\widetilde{\mathbf{B}}^k(\widetilde{\mathbf{V_0}})}\left|\widetilde{\mathbf{F}}\cdot\mathbf{v}\right|$$
 and thus 
 $$\rho(\mathbf{B},\mathbf{V_0},\mathbf{F})=\rho(\widetilde{\mathbf{B}},\widetilde{\mathbf{V_0}},\widetilde{\mathbf{F}}).$$
\end{lemma}
This lemma tells us that one can assume that every coordinate is accessible and co-accessible without lose of generality.
Remark, that in general the set of accessible and co-accessible coordinates can easily be computed by a recursive algorithm. Moreover, if our bilinear map is associated to a minimal tree-automaton whose garbage state was removed then the bilinear map is already accessible and co-accessible. We are now ready to discuss the computability of the growth rate of accessible and co-accessible bilinear systems.

For any set of points $X$, we denote by $\operatorname{conv}(X)$ the convex hull of $X$.
\begin{lemma}\label{polytopefirstdirection}
Let $n$ be an integer, $\mathbf{B}:\mathbb{R}^{n}\times\mathbb{R}^{n}\rightarrow\mathbb{R}^{n}$ be a  bilinear map and $\mathbf{V_0},\mathbf{F}\in\mathbb{R}^{n}$ be vectors.
Suppose that there is a bounded set of vectors $X\subseteq\mathbb{R}^n$  such that:
\begin{itemize}
 \item $\mathbf{V_0}\in\operatorname{conv}(X)$,
 \item $\forall \vect{u},\vect{v}\in X$, we have   $B(\vect{u},\vect{v})\in\operatorname{conv}(X)$.
\end{itemize}
Then $\rho(\mathbf{B},\mathbf{V_0},\mathbf{F})\le1$.
More precisely, for all $\vect{u}\in \bigcup_{k\ge1}\mathbf{B}^k(\mathbf{V_0})$, $|\mathbf{F}\cdot \vect{u}| \le \sup_{ x\in X}|\mathbf{F}\cdot x|$.
\end{lemma}
\begin{proof}
The proof mostly relies on trivial manipulations of convex sets.
First, remark that 
\begin{equation}\label{convremark}
 \operatorname{conv}(\{\mathbf{B}(x,y): x,y\in X\})\subseteq\operatorname{conv}(X).
\end{equation}

Let us first show by induction on $k$ that for all $\vect{u}\in \mathbf{B}^k(\mathbf{V_0})$,
$\vect{u}\in  \operatorname{conv}(X)$.
If $\vect{u}\in\mathbf{B}^1(\mathbf{V_0})$, then by definition $\vect{u}=\mathbf{V_0}\in\operatorname{conv}(X)$.

Suppose $\vect{u}=\mathbf{B}(\vect{u_1},\vect{u_2})$ with $\vect{u_1}\in\mathbf{B}^i(\mathbf{V_0})$, $\vect{u_2}\in \mathbf{B}^j(\mathbf{V_0})$ and $i+j=n$.
 By induction hypothesis there are two functions $f_1, f_2:X\mapsto[0,1]$ such that:
 $\sum_{x\in X} f_1(x)=\sum_{x\in X} f_2(x)=1$, $\vect{u_1}=\sum_{x\in X} f_1(x) x$ and $\vect{u_2}=\sum_{x\in X} f_2(x) x$.
 By bilinearity of $\mathbf{B}$ we get:
$$  \vect{u}=\mathbf{B}(\vect{u_1},\vect{u_2})
    =\sum_{x\in X}\sum_{y\in X}f_1(x)f_2(y)\mathbf{B}(x,y)$$
  Moreover, $\sum_{x\in X}\sum_{y\in X}f_1(x)f_2(y)=1$ implies that $\vect{u}\in\operatorname{conv}(\{\mathbf{B}(x,y): x,y\in X\})$ which implies  $\vect{u}\in\operatorname{conv}(X)$ by equation \eqref{convremark}.  

Now we know that for all $k$ and all $\vect{u}\in \mathbf{B}^k(\mathbf{V_0})$,
$\vect{u}\in  \operatorname{conv}(X)$. 
It implies that there is a function $f:X\mapsto[0,1]$ such that:
 $\sum_{x\in X} f(x)=1$ and $\vect{u}=\sum_{x\in X} f(x) x$.
Thus $|\mathbf{F}\cdot\vect{u}|=|\sum_{x\in X} f(x) \mathbf{F}\cdot x| \le \sup_{ x\in X}|\mathbf{F}\cdot x|$.
Since $X$ is a bounded set this implies that there is a constant $C$ such that for all $k$, 
$$\max_{\mathbf{v}\in\mathbf{B}^k(\mathbf{V_0})}\left|\mathbf{F}\cdot\mathbf{v}\right| \le C$$
and thus
$$\rho(\mathbf{B},\mathbf{V_0},\mathbf{F})\le1$$
which concludes the proof.
\end{proof}

If one can find such a set $X$ then we deduce that $\rho(\mathbf{B},\mathbf{V_0},\mathbf{F})\le1$. On the other hand, Lemma \ref{converse} tells us that if $\rho(\mathbf{B},\mathbf{V_0},\mathbf{F})<1$ then there is such a set $X$. We first state an intermediate Lemma about the approximation of convex sets that will be central in the proof of Lemma \ref{converse}.
For any set $S\subseteq\mathbb{R}^n$ and any real $y$, we let $\frac{S}{y}=\left\{\frac{\mathbf{v}}{y}: \mathbf{v}\in S\right\}$.

\begin{lemma}\label{nestedpolytopes}
Let $S\subseteq\mathbb{R}^n$ be a bounded convex set and $\varepsilon\in\mathbb{R}_{>0}$ be a positive constant. If $0$ belongs to the interior of $S$ there exists
a finite set $X\subseteq\mathbb{Q}^n$ such that
$$\frac{S}{1+\varepsilon}\subseteq\operatorname{conv}(X)\subseteq S.$$
\end{lemma}
\begin{proof}
Let $x\in\mathbb{R}^n$ be any point of the boundary of $\frac{S}{1+\varepsilon}$.
Since $0$ belongs to the interior of $S$ there exists $\beta\in\mathbb{R}_{>0}$ such that $\mathcal{B}$ the
closed ball of radius $\beta$ centered at $0$ is contained in $S$.
Since $(1+\varepsilon)x\in S$ then the convex hull $C$ of $(1+\varepsilon)x$ and $\mathcal{B}$ is included in $S$. Thus the minimum distance between $x$ and the boundary of $S$ is at least the distance between $x$ and the boundary of $C$. 
Since the center of $\mathcal{B}$, $x$ and $(1+\varepsilon)x$ are on the same line, the ball $\mathcal{B}'$ of center $x$ and of radius $\frac{\beta\epsilon}{1+\epsilon}$ is contained in $C$ (apply the intercept theorem to any radius of $\mathcal{B}'$ and the parallel radius of $\mathcal{B}$). Thus the distance between the boundary of $S$ and the boundary of $\frac{S}{1+\varepsilon}$ is at least $\frac{\beta\epsilon}{1+\epsilon}$.

Let $l\in\mathbb{Q}_{>0}$ be a rational such that $l<\frac{\beta\epsilon}{\sqrt{n}(1+\epsilon)}$ and let $X$ be the intersection of the grid of step $l$ with $S$, that is $$X=S\cap\{lx: x\in \mathbb{Z}^n\}.$$
Then the set $X$ is a finite set of $\mathbb{Q}^n$.
Moreover, the largest diagonal of a cell is smaller than $\frac{\beta\epsilon}{1+\epsilon}$, this implies that for any point $x\in\frac{S}{1+\varepsilon}$ the cells that contain $x$ are all fully contained in $S$ and so are the vertices of theses cells. This implies that any point from $\frac{S}{1+\varepsilon}$ is in the convex hull of $X$, which concludes the proof.
\end{proof}

\begin{lemma}\label{converse}
 Let $n$ be an integer, $\mathbf{B}:\mathbb{R}_{\ge 0}^{n}\times\mathbb{R}_{\ge 0}^{n}\rightarrow\mathbb{R}_{\ge 0}^{n}$\footnote{This is the first proof where we use the non-negativity condition. This lemma might hold even without this condition.} be a  bilinear map and $\mathbf{V_0},\mathbf{F}\in\mathbb{R}_{\ge 0}^{n}$ be vectors such that all the coordinates are accessible and co-accessible. Suppose that $\rho(\mathbf{B},\mathbf{V_0},\mathbf{F})<1$ then there is a finite set of rational vectors $X\subseteq\mathbb{Q}^n$ such that:
\begin{itemize}
 \item $\mathbf{V_0}\in\operatorname{conv}(X)$,
 \item for all $\vect{u},\vect{v}\in X$, we have $B(\vect{u},\vect{v})\in\operatorname{conv}(X)$.
\end{itemize}
\end{lemma}
\begin{proof}
Let $n$ be an integer, $\mathbf{B}:\mathbb{R}_{\ge 0}^{n}\times\mathbb{R}_{\ge 0}^{n}\rightarrow\mathbb{R}_{\ge 0}^{n}$ be a  bilinear map and $\mathbf{V_0},\mathbf{F}\in\mathbb{R}_{\ge 0}^{n}$ be vectors such that all the coordinates are accessible and co-accessible with $\rho(\mathbf{B},\mathbf{V_0},\mathbf{F})<1$.
Let $\varepsilon\in\mathbb{R}_{>0}$ be a positive real such that 
$(1+\varepsilon)^2\rho(\mathbf{B},\mathbf{V_0},\mathbf{F})<1$. Then by Lemma \ref{simpleBilinearity}, we get $\rho((1+\varepsilon)\mathbf{B},(1+\varepsilon)\mathbf{V_0},\mathbf{F})<1$.

For any $x,y\in {\mathbb{R}}^n$, we write $x\le y$ if each coordinate of $x$  is at most equal to the corresponding coordinate of $y$.
For any $X\subseteq{\mathbb{R}_{\ge 0}}^n$, let $\operatorname{conv}_\le(X)= \{\vect{x}\in\mathbb{R}_{\ge 0}^n: \exists\vect{x'}\in \operatorname{conv}(X), \vect{x}\le\vect{x'}\}$ and 
$$S=\operatorname{conv}_\le\left(\bigcup_{k\ge0}((1+\varepsilon)\mathbf{B})^k(\mathbf{(1+\varepsilon)V_0})\right)$$

If $S$ is unbounded then so is $\bigcup_{k\ge0}((1+\varepsilon)\mathbf{B})^k(\mathbf{(1+\varepsilon)V_0})$.
Since all the coordinates are co-accessible, if that would imply that $\max_{\mathbf{v}\in((1+\varepsilon)\mathbf{B})^k((1+\varepsilon)\mathbf{V_0})}\left|\mathbf{F}\cdot\mathbf{v}\right|$ is unbounded which contradicts the fact that $\rho((1+\varepsilon)\mathbf{B},(1+\varepsilon)\mathbf{V_0},\mathbf{F})<1$. Thus all the coordinates of any element of $S$ are bounded. 

By bilinearity of $(1+\varepsilon)\mathbf{B}$ and the fact that all the coordinates of $\mathbf{B}$ are non-negative, for any $x,y\in S$, $(1+\varepsilon)\mathbf{B}(x,y)\in S$ and thus $\mathbf{B}(x,y)\in  \frac{S}{1+\varepsilon}$. Moreover, $(1+\varepsilon)\mathbf{V_0}\in S$ and thus
$\mathbf{V_0}\in  \frac{S}{1+\varepsilon}$.
We deduce, that for any set $X$ such that $\frac{S}{1+\varepsilon}\subseteq \operatorname{conv}(X)\subseteq S$, we get
\begin{itemize}
 \item $\mathbf{V_0}\in\operatorname{conv}(X)$,
 \item for all $\vect{u},\vect{v}\in X$, we have $\mathbf{B}(\vect{u},\vect{v})\in\operatorname{conv}(X)$.
\end{itemize}

In order to conclude the proof we only need to show that there exists a finite set 
$X\subseteq\mathbb{Q}^n$ such that $\frac{S}{1+\varepsilon}\subseteq \operatorname{conv}(X)\subseteq S$, but we cannot apply Lemma \ref{nestedpolytopes} immediately since $0$ is in the boundary of $S$.  

Let $S'= \operatorname{conv}(S\cup \{-\mathds{1}\})$ where $\mathds{1}$ is the vector whose all coordinates are $1$.
Since all the coordinates are accessible there exists a positive real $\alpha\in\mathbb{R}_{\ge0}$ such that for all $i\in\{1,\ldots, n\}$, the vector $\alpha \vect{e}_i$ belongs to $S$ ($\alpha \vect{e}_i$ is the vector whose $i$th coordinate is set to $\alpha$ and every other coordinate to 0).
Since $S'$ also contains $-\mathds{1}$ and is convex, $0$ is in the interior of $S'$. Thus by Lemma \ref{nestedpolytopes}, we deduce that there exists
 a finite set $X'\subseteq\mathbb{Q}^n$ such that
$$\frac{S'}{1+\varepsilon}\subseteq\operatorname{conv}(X')\subseteq S'.$$
Since $S=\operatorname{conv}_{\le}(S)$, we get $S'\cap \mathbb{R}_{\le0}=S$ and thus $$\frac{S}{1+\varepsilon}\subseteq \operatorname{conv}(X')\cap \mathbb{R}_{\le0}^n\subseteq S.$$
Moreover since $X'\subseteq\mathbb{Q}^n$ is a finite set, there exists a finite set $X\subseteq\mathbb{Q}^n$ such that $\operatorname{conv}(X)= \operatorname{conv}(X')\cap \mathbb{R}_{\le0}^n$. We finally get
$$\frac{S}{1+\varepsilon}\subseteq \operatorname{conv}(X)\subseteq S$$
which concludes the proof.
\end{proof}

As a Corollary of Lemma \ref{polytopefirstdirection} and Lemma \ref{converse} we get.

\begin{corollary}\label{polytopesaregoodenough}
 Let $n$ be an integer, $\mathbf{B}:\mathbb{R}_{\ge 0}^{n}\times\mathbb{R}_{\ge 0}^{n}\rightarrow\mathbb{R}_{\ge 0}^{n}$ be a  bilinear map and $\mathbf{V_0},\mathbf{F}\in\mathbb{R}_{\ge 0}^{n}$ be vectors such that all the coordinates are accessible and co-accessible. 
 Let $\Lambda\subseteq\mathbb{R}_{\ge 0}$ be the largest set such that for all $\lambda\in\Lambda$, there exists a finite set of rational vectors $X\subseteq\mathbb{Q}^n$ such that:
\begin{itemize}
 \item $\frac{\mathbf{V_0}}{\lambda}\in\operatorname{conv}(X)$,
 \item for all $\vect{u},\vect{v}\in X$, we have $B(\vect{u},\vect{v})\in\operatorname{conv}(X)$.
\end{itemize}
Then $\rho(\mathbf{B},\mathbf{V_0},\mathbf{F})= \inf \Lambda$.
\end{corollary}
This corollary holds if we replace the condition $\Lambda\subseteq\mathbb{R}$ by $\Lambda\subseteq \mathbb{S}$
where $\mathbb{S}$ is dense in $\mathbb{R}$ (this is particularly interesting if we take the set of rational number or the set of algebraic numbers).
We finally deduce that the growth rate of a bilinear system if approximable from above.
\begin{theorem}\label{mainTh}
 There exists an algorithm that takes as input an integer $n$, a bilinear map $\mathbf{B}:\mathbb{R}_{\ge 0}^{n}\times\mathbb{R}_{\ge 0}^{n}\rightarrow\mathbb{R}_{\ge 0}^{n}$ and two vectors $\mathbf{V_0},\mathbf{F}\in\mathbb{R}_{\ge 0}^{n}$ and that approximate from above $\rho(\mathbf{B},\mathbf{V_0},\mathbf{F})$. That is, this algorithm outputs a decreasing sequence of rational numbers $(a_i)_{i\in\mathbb{N}}$ such that
 $\lim\limits_{i\mapsto\infty}a_i= \rho(\mathbf{B},\mathbf{V_0},\mathbf{F})$.
\end{theorem}
\begin{proof}
 Given a rational number $\alpha\in\mathbb{Q}$ one can enumerate all the finite sets $X\subseteq\mathbb{Q}^n$ until finding such a set that respects the conditions of Lemma \ref{polytopefirstdirection} with $\mathbf{B}$,$\frac{\mathbf{V_0}}{\alpha}$ and $\mathbf{F}$.
 If $\rho(\mathbf{B},\mathbf{V_0},\mathbf{F})<\alpha$, Lemma \ref{converse} implies that we will find such a set in which case we will be able to deduce that $\rho(\mathbf{B},\mathbf{V_0},\mathbf{F})\le\alpha$ by Lemma \ref{polytopefirstdirection}.
 
 Now,we can do that for all the rationals ``in parallel'' (enumerate the rational and, between each new rational, run one step of computation for every rational already enumerated). Whenever the search succeeds for a rational $\alpha$, we have a new upper bound for $\rho(\mathbf{B},\mathbf{V_0},\mathbf{F})\le\alpha$ and we can stop running the search for every larger rational. 
 This sequence of upper bound is clearly increasing and since the search is going to succeed at some point for any valid upper bound the sequence converges toward the best upper bound, that is $\rho(\mathbf{B},\mathbf{V_0},\mathbf{F})$.
\end{proof}
This is obviously not possible to use this algorithm even for the simplest cases. However, in the next section we are able to use this approach anyway to deduce a bunch of results.

\section{Applications}\label{results}
As already stated, the algorithm of Theorem \ref{mainTh} cannot be use to obtain interesting bounds.
Instead of trying all convex bodies, it is much more efficient to explicitly compute  $\operatorname{conv}_\le\left(\bigcup_{k=1}^n\mathbf{B}^k\left(\frac{\mathbf{V_0}}{\alpha}\right)\right)$\footnote{Remark, that if one find a set $X$ such that $\operatorname{conv}(X)$ is as desired then so is $\operatorname{conv}_{\le}(X)$, but using $\operatorname{conv}_{\le}$ leads to more efficient computations (at least when one uses linear programming to find $\operatorname{conv}(X)$ or $\operatorname{conv}_{\le}(X)$).} for some ``small''\footnote{By small, we  mean ``as large as computationally feasible'' which tends to be small.} $n$ and a chosen $\alpha\in\mathbb{R}_{\ge0}$.
In some cases, it reaches the limit in finitely many steps. 
In other cases,  by inspecting $\operatorname{conv}_\le\left(\bigcup_{k=1}^n\mathbf{B}^k\left(\frac{\mathbf{V_0}}{\alpha}\right)\right)$, one can guess the limit or a set that would respect the conditions.
In many cases, it is easy to find construction that give lower bounds.
In fact, by inspecting the vertices of $\operatorname{conv}_\le\left(\bigcup_{k=1}^n\mathbf{B}^k\left(\frac{\mathbf{V_0}}{\alpha}\right)\right)$ one can find the corresponding extreme construction.

Instead of explicitly computing $\operatorname{conv}_\le\left(\bigcup_{k=1}^n\mathbf{B}^k\left(\frac{\mathbf{V_0}}{\alpha}\right)\right)$ for all $k$ it is slightly more efficient to use the following algorithm.

\begin{figure}[h!]
\begin{algorithm}[H]\label{algofindpolytope}
 \KwData{$\mathbf{B}, \vect{v}_0$, $\alpha$}
 \KwResult{A set $X$ that respects the conditions of Lemma \ref{polytopefirstdirection} with $\mathbf{B}$ and $\frac{\vect{v}_0}{\alpha}$}
 $X:=\{\frac{\vect{v}_0}{\alpha}\}$\;
  \While{$\exists \vect{x},\vect{y}\in X$  such that $\mathbf{B}(\vect{x},\vect{y})\not\in \operatorname{conv}_\le(X)$}{
    $X':=\{\mathbf{B}(\vect{x},\vect{y}): \mathbf{B}(\vect{x},\vect{y})\in X\}$\;
    $X := \operatorname{Hull}_\le(X\cup X')$\;
 }
\caption{Computation of the set $X$}
\end{algorithm}
\end{figure}

In this algorithm, for all $Y$, $\operatorname{Hull}_\le(Y)$ is the smallest subset $Y'$ of $Y$ that verifies $\operatorname{conv}_\le(Y')=\operatorname{conv}_\le(Y)$. Remark, that $\operatorname{Hull}_\le$ can be easily computed using linear programming. Most of the computation time is spent on $\operatorname{Hull}_\le$ and it would be interesting to reduce it (even our simplex implementation is sub-optimal).

Our implementation of this technique uses Mona \cite{monamanual2001} to obtain the tree automaton from an $MSO_1$ formula. Then a C++ program reads the output and produces the bilinear map. Our program then requires the user to input a value $\alpha$ (either as a rational or as the root of a polynomial) and inductively computes $\operatorname{conv}_\le\left(\bigcup_{k=1}^n\mathbf{B}^k\left(\frac{\mathbf{V_0}}{\alpha}\right)\right)$ until reaching a fix-point. The user can also provide some other vectors that should be added to the set, this is useful whenever the limit of $\operatorname{conv}_\le\left(\bigcup_{k=1}^n\mathbf{B}^k\left(\frac{\mathbf{V_0}}{\alpha}\right)\right)$ is not reached in finite time but can be guessed by the user. We only use exact computations either on the rationals or on finite algebraic extension of $Q$ so there is no issue of precision.
More implementation details are given in Annex \ref{C++desc}.

In this Section, we first study the example of independent dominating sets in details in order to illustrate the technique presented here. This is a nice example since the computations fit in a human brain. Then we provide other results where the aid of the computer is necessary. 
The choice of examples is arbitrary and the main criterion of choice is that these sets were interesting enough to be named. 
It should be noted that we obtained good bounds for every familly of sets that we tryed and that they are all listed in this section (that is, we don't know yet of a familly of set definable in $MSO_1$ for which our approach fails completely). 
Recall that the case of minimal dominating set was already solved by Rote that gave the upper bound of $95^\frac{1}{13}$, our program can verify and agrees with the result of Rote \cite{r-mnmds-19,arxivRote}. 

\subsection{Independent dominating sets}
A set of vertices $S$ of a graph $(V,E)$ is an independent dominating set if it is an independent set and a dominating set. That is, no two vertices of $S$ are adjacent and every vertex outside of $S$ has a neighbor in $S$.
Moreover, a set is an independent dominating sets if and only if is is a maximal independent sets if and only if its complement is a minimal vertex-cover (or a  minimal transversal). 
It was showed in \cite{indepDom}, that the maximal number of independent dominating set of a tree of order $n$ is $2^\frac{n-1}{2}$ when $n$ is odd, and is $2^{\frac{n}{2}-1}+1$ when $n\ge2$ is even. 
We provide an alternative proof of this result in this subsection to illustrate our approach.

Independent dominating sets can be defined by the following formula.
\begin{align*}
\Psi(S):=&\left(\forall x,y: (x\in S\land y\in S)\implies \lnot E(x,y)\right)\\
\land&\left(\forall x: x\notin S \implies \exists y: y\in S \land E(x,y)\right).
\end{align*}

The associated tree automaton is $\mathcal{A}=(Q, \Sigma, Q_f, \delta)$ where $\Sigma=\{J_0,J_1,\basetree_0,\basetree_1\}$, $Q=\{F,D,d\}$, $Q_f=\{D,d\}$, $\delta$ is such that
\begin{align*}
\delta(\basetree_0)=F \quad\quad\quad\quad\quad& \delta(\basetree_1)=D\\
\delta(J_0,F,D)=d \quad\quad\quad\quad\quad& \delta(J_0,F,d)=F\\
\delta(J_0,D,d)=D \quad\quad\quad\quad\quad&  \delta(J_0,D,F)=D \\
\delta(J_0,d,D)=d \quad\quad\quad\quad\quad& \delta(J_0,d,d)=d 
\end{align*}
and $\delta$ is not defined in the remaining cases. Recall that there is no transition for $J_1$ since the set $S$ is consistent.

This tree automaton is automaticaly obtained using our implementation, but one could easily find this automaton without knowing links between $MSO$ and automata. Indeed, one may interpret $D$ as ``the current root is in the independent dominating set'', $d$ as ``the current root is not in the independent set and is already dominated by one of its children'' and $F$ as ``the current root is not in the independent set and is not already dominated by one of its children''. Then clearly two nodes with the state $D$ should not be connected since this wouldn't give an independent set and thus there is no transition $\delta(J_0,D,D)$. If one connects a node in state $F$ with a node in state $d$, then the nodes are still in the same state and the node in state $F$ has to be the new root otherwise it will not gain any new neighbor in the future to dominate it, thus $\delta(J_0,F,d)=F$ and there is no transition $\delta(J_0,d,F)$. If one connects a node in state $F$ with a node in state $D$, then the node in state $F$ is now in state $d$ and we get $\delta(J_0,F,D)=d$ and $\delta(J_0,D,F)=D$. The other transitions can be deduced in a similar way.

The bilinear map $\mathbf{B}:\mathbb{R}_{\ge 0}^{3}\times\mathbb{R}_{\ge 0}^{3}\rightarrow\mathbb{R}_{\ge 0}^{3}$ and  vectors $\mathbf{V_0},\mathbf{F}\in\mathbb{R}_{\ge 0}^{3}$ that can be computed from $\Psi$ that respect Corollary \ref{formulatolinearalg} are given by:
$$\mathbf{V_0}=
\begin{pmatrix}
1\\1\\0
\end{pmatrix},\,
\mathbf{F}=
\begin{pmatrix}
0\\1\\1
\end{pmatrix}
\text{ and for all } x,y\in\mathbb{R}^3,\,
\mathbf{B}(x,y)=
\begin{pmatrix}
x_1y_3\\
x_2y_1+x_2y_3\\
x_1y_2+x_3y_3+x_3y_2
\end{pmatrix}\,. $$

Let $X$ be the set $X=\left\{
\begin{pmatrix}
0\\0\\\frac{1}{\sqrt{2}}
\end{pmatrix},
\begin{pmatrix}
0\\\frac{1}{2}\\\frac{1}{2}
\end{pmatrix},
\begin{pmatrix}
\frac{1}{\sqrt{2}} \\\frac{1}{\sqrt{2}} \\0
\end{pmatrix}
\right\}$. 
One easily checks that $X$ is stable by $\mathbf{B}$ and contains $\frac{\mathbf{V_0}}{\sqrt{2}}$.
We deduce by Lemma \ref{polytopefirstdirection}, that $\rho(\mathbf{B},\frac{\mathbf{V_0}}{\sqrt{2}},\mathbf{F})\le1$ which implies  by Lemma \ref{simpleBilinearity} that $$\rho(\mathbf{B},\mathbf{V_0},\mathbf{F})\le\sqrt{2}\,.$$
Thus if we let $\rho_n$ be the maximal number of independent dominating set of a tree of order $n$, we deduce that
$\lim\limits_{n\rightarrow\infty}(\rho_n )^\frac{1}{n}=\sqrt2$. In fact, Lemma \ref{polytopefirstdirection} even tells us that for all $n$, $\rho_n\le  2^{\frac{n}{2}} \max_{x\in X} \mathbf{F}\cdot x= 2^{\frac{n}{2}}$ .

For any positive integer $n$, let $T_n$ be the graph given by (see Figure \ref{indepdomtree}) 
$$T_n=(\{s,s_1,\ldots,s_{2n}\}, \{(s,s_{2i-1}): i\in{1,\ldots,n}\}\cup\{(s_{2i-1}, s_{2i}): i\in{1,\ldots,n}\}).$$
Let $D$ be a set such that for all $i$, $|S\cap\{s_{2i},s_{2i+1}\}|=1$ and $s\in D$ if and only if for all $i$, $s_{2i-1}\in S$. 
Then $D$ is an independent dominating set of $T_n$. Thus there are at least $2^{n}$ independent dominating sets of $T_n$ which is of order $2n+1$.

\begin{figure}[H]
\centering
\includegraphics{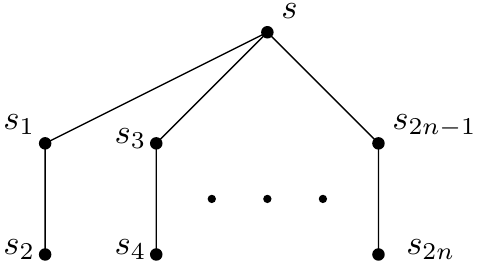}
\caption{Trees with $\Theta(2^{\frac{n}{2}})$ independent dominating sets \label{indepdomtree}}
\end{figure}
\begin{proposition}\label{indepdom}
The number of independent dominating sets in a tree of order $n$ is at most $2^{\frac{n}{2}}\approx 1.4142135^n$.
It is sharp in the sense that the number of independent dominating sets is at least $2^{\frac{n-1}{2}}=\frac{ 2^{\frac{n}{2}}}{\sqrt{2}}$ for infinitely many trees.
\end{proposition}
In fact, the gap between the upper-bound and the lower-bound is only due to the paths on $2$ and $4$ vertices and we can provide an upper bound that is reached for every odd integer.
\begin{proposition}
Let $n\ge5$. The number of independent dominating sets in a tree of order $n$ is at most $ 2^{\frac{n-1}{2}}$. 
\end{proposition}
\begin{proof}
This result can be obtained by considering the $MSO_1$ property ``$S$ is an independent dominating set and the tree is of order at least $5$''. The associated tree automaton is bigger, but Algorithm \ref{algofindpolytope} still converges quickly toward the limit convex set. One easily deduces a set $X$ that respects the conditions of Lemma \ref{polytopefirstdirection} (such a set is provided in our implementation). The bound is then simply obtained from application of Lemma  \ref{polytopefirstdirection} on this set.
\end{proof}
We were able to provide the sharp bound for every odd integer, but the bound we got for the odd integers is not sharp since according to \cite{indepDom} it should be $2^{\frac{n}{2}-1}+1$.
We could in fact use exactly the same technique for the even case. Indeed, the property ``the tree is of even order'' cannot be expressed in $MSO_1$, but is recognizable by tree-automatons. However, the tools that we used did not allow us to easily implement this.

This trick to obtain better asymptotic bounds is worth mentioning, but we did not used it with on the other results of this section. It is however, clear that this would give better bounds for all the results where the multiplicative constant is not optimal. 

\subsection{Total perfect dominating sets}
A subset of vertices of a tree is a total perfect dominating set if
every vertices have exactly one neighbor in this set. The initial vector, final vector and the bilinear map corresponding to these sets are respectively given by
$$\mathbf{V_0}=
\begin{pmatrix}
0\\0\\1\\1
\end{pmatrix},\,
\mathbf{F}=
\begin{pmatrix}
1\\1\\0\\0
\end{pmatrix}
\text{ and for all } u,v\in\mathbb{R}^4,\,
\mathbf{B}(x,y)=
\begin{pmatrix}
 u_1v_1+ u_3v_2\\
 u_2v_3+ u_4v_4\\
 u_3v_1\\
 u_4v_3
\end{pmatrix}\,. $$

We get the following result
\begin{proposition}\label{treetpd}
Let $\alpha = (2^{27}\times 7)^\frac{1}{85}\approx1.275157$ and $C=\frac{\alpha^{80}}{234881024}\approx 1.186429$.
The number of total perfect dominating sets of a tree of order $n$ is upper-bounded by $C\alpha^n$.
This value of $\alpha$ is sharp.
\end{proposition}
\begin{proof}
The computation of $\bigcup_{k\ge1}\mathbf{B}^k(\frac{\mathbf{V_0}}{\alpha})$ by Algorithm \ref{algofindpolytope} converges in finitely many steps and we reach a finite set $X$ such that $\operatorname{conv}(X)=\operatorname{conv}\left(\bigcup_{k\ge1}\mathbf{B}^k(\frac{\mathbf{V_0}}{\alpha})\right)$. By construction this set respects the conditions of Lemma \ref{polytopefirstdirection} which implies that for all  $\vect{u}\in \bigcup_{k\ge1}\mathbf{B}^k(\mathbf{V_0})$, $|\mathbf{F}\cdot \vect{u}| \le \max_{ x\in X}|\mathbf{F}\cdot x|$. Given the set $X$ one easily verifies $C=\max_{ x\in X}|\mathbf{F}\cdot x|$ which concludes the proof of the upper bound  $C\alpha^n$. The set $X$ is provided in Annex \ref{tpdset}.

Let $T$ be the tree depicted in figure \ref{treetpds}.
\begin{figure}
\centering
\includegraphics[scale=1]{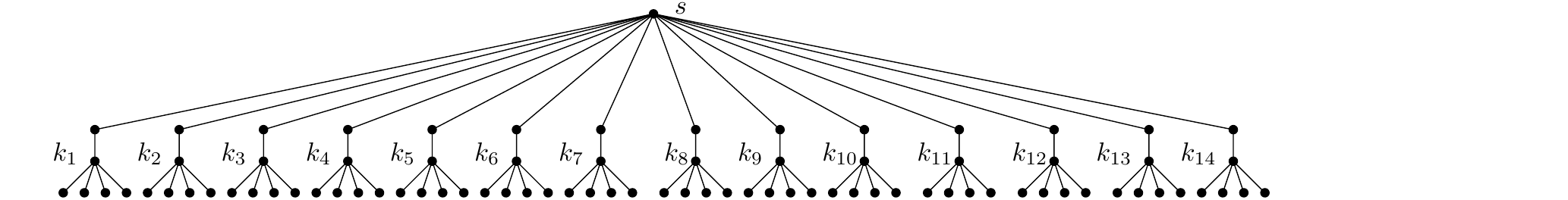}
\caption{The tree $T$.\label{treetpds}}
\end{figure}
\begin{figure}
\centering
\includegraphics[scale=0.75]{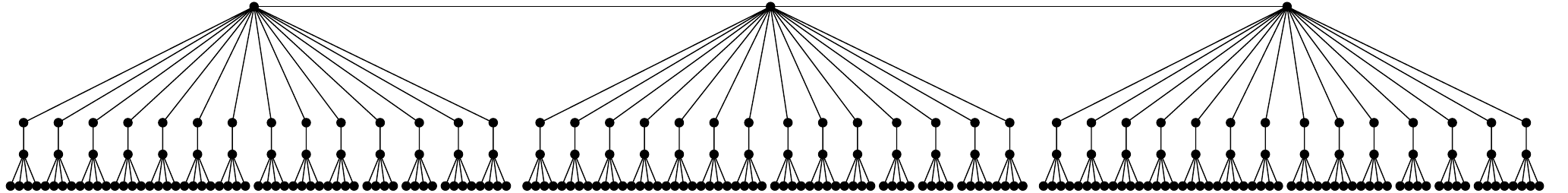}
\caption{The tree $T_3$ that has $(2^{27}\times 7)^3$ total perfect dominating sets.\label{treetpds2}}
\end{figure}
Let $D$ be a perfect total dominating set of $T$.  
Each leaf of the tree needs a neighbor in $D$, so all the $k_i$ are in $D$.
This implies that $s$ is not in $D$, otherwise the nodes between $s$ and the $k_i$ would have two neighbors in $D$.
Moreover, exactly one neighbor of $s$ is in $D$ and $k_i$ is the only other vertex in $D$ in the corresponding subtree, since $k_i$ needs exactly one neighbor in $D$.
For the other subtrees, exactly one of the leaf of each of them is in $D$.
Thus there are at least $(4^{13}\times 14)=2^{27}\times 7$ perfect total dominating sets.

Now let $T_k$ be a chain of $k$ copies of $T$ where two consecutive copies share an edge between their roots.
Then since the roots cannot be in the perfect total dominating sets the number of perfect total dominating sets is exactly $(2^{27}\times 7)^k= \alpha^{85k}$.
Thus $T_k$ has $n=k (6\times14+1)=85k$ vertices and $\alpha^n$ total perfect dominating sets.
\end{proof}

This case illustrates the power of our technique. Given that the growth rate is $(2^{27}\times 7)^\frac{1}{85}$, it seems that there should be no elementary proof. However, the set $X$ is found automatically without human intervention in approximately 10 minutes on an average laptop. The construction of the lower bound corresponds to one of the vertices obtained when trying to find $X$ with a smaller $\alpha$.

\subsection{Perfect codes}
A subset $S$ of vertices of a graph $(V,E)$ is a perfect code if any vertex in $V\setminus S$ has exactly one neighbor in $S$ and any vertex of $S$ has no neighbor in $S$. In other words, a set is a perfect code if it is an independent set and a perfect dominating set.
The initial vector, final vector and the bilinear map corresponding to these sets are respectively given by
$$\mathbf{V_0}=
\begin{pmatrix}
0 \\ 1\\  1 
\end{pmatrix},\,
\mathbf{F}=
\begin{pmatrix}
1\\  1\\    0
\end{pmatrix}\text{ and for all } u,v\in\mathbb{R}^3,\,
\mathbf{B}(x,y)=
\begin{pmatrix}
  u_1v_1+ u_3v_2\\
   u_2v_3\\
   u_3v_1
\end{pmatrix}\,. $$

\begin{proposition}
Let $\alpha = 3^{\frac{1}{7}}\approx1.16993$ and  $C=\frac{2}{3}\alpha^5\approx  1.4612$. 
The number of perfect codes of a tree of order $n$ is upper-bounded by $C\alpha^n$.
This value of $\alpha$ is tight.
\end{proposition}
\begin{proof}
The proof that the upper-bound is correct can be done by finding a set $X$ that respects the conditions of Lemma \ref{polytopefirstdirection}. One easily verifies that this is the case for the following set that can be found by Algorithm \ref{algofindpolytope} 
$$X=\left\{
\begin{pmatrix}
 0 \\  \frac{1}{3}\alpha^6 \\  \frac{1}{3}\alpha^6   
\end{pmatrix},
\begin{pmatrix}
\frac{1}{3}\alpha^4 \\  0 \\  \frac{1}{3}\alpha^4   
\end{pmatrix},
\begin{pmatrix}
\frac{1}{3}\alpha^5  \\ \frac{1}{3}\alpha^5 \\  0  
\end{pmatrix},
\begin{pmatrix} 
\frac{2}{3}\alpha^2  \\ 0 \\  \frac{1}{3}\alpha^2  
\end{pmatrix},
\begin{pmatrix} 
1  \\ 0 \\  \frac{1}{3}   
\end{pmatrix}\right\}$$

In order to verify that this bound is sharp, let $T$ be the tree made of a central vertex $s$ with $3$ pendant $P_2$ (this graph is given in Figure \ref{treepc}).
Let $T_k$ be the tree made of $k$ copies of $T$ connected by a path going through the central vertices.
\begin{figure}[H]
\centering
\includegraphics{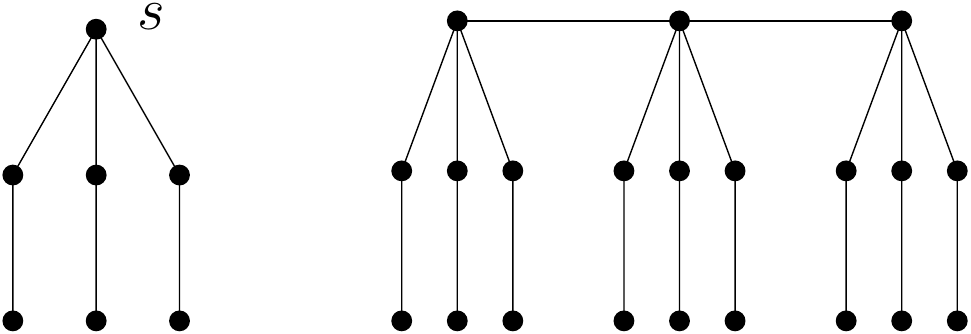}
\caption{On the left $T$ and on the right $T_3$.\label{treepc}}
\end{figure}
The number of perfect codes of $T$ is $3$ and none of the perfect codes contain $s$, thus $T_k$ at least $3^{k}$ perfect codes. Since $T_k$ has $7k$ vertices this concludes the proof of the lower bound.
\end{proof}

\subsection{Minimal perfect dominating set}
A subset $S$ of vertices of a graph $(V,E)$ is a perfect dominating set if every vertex in $V\setminus S$ has exactly one neighbor in $S$.
A perfect dominating set $S$ is minimal if there is no other perfect dominating set strictly included in $S$.

The star on $n$ vertices has $2^{n-1}-1$ perfect dominating sets, so the growth rate of perfect dominating sets over trees is $2$. However, the case of minimal perfect dominating sets is not so trivial.
The initial vector, final vector and the bilinear map corresponding to these sets are respectively given by
$$\mathbf{V_0}=
\begin{pmatrix}
0 \\ 0 \\ 1\\  1\\  0\\  0 
\end{pmatrix},\,
\mathbf{F}=
\begin{pmatrix}
1\\  1\\  1\\  0\\  0\\  0  
\end{pmatrix}$$
$$\text{ and for all } u,v\in\mathbb{R}^6,\,
\mathbf{B}(x,y)=
\begin{pmatrix}
 u_1v_1+ u_1v_4+ u_1v_5+ u_3v_4+ u_5v_1+ u_5v_4+ u_6v_4\\
 u_2v_2+ u_4v_1+ u_4v_3+ u_4v_5+ u_4v_6\\
0\\
 u_4v_2\\
 u_3v_1+ u_5v_5+ u_6v_1\\
 u_3v_5+ u_6v_5\\
\end{pmatrix}\,. $$
We can show the following result.
\begin{proposition}\label{minperfectdomsets}
 Let $\alpha$ be the real root of $x^3-x-1$ between 1 and 2, $\alpha\approx1.32472$ and $C=- 2\alpha^2 + 2\alpha + 2\approx1.14133$. 
 Then the number of minimal perfect dominating set in a tree of order $n$ is bounded by $C\alpha^n$.
 Moreover, this value of $\alpha$ is sharp even for paths.
\end{proposition}
\begin{proof} 
Algorithm \ref{algofindpolytope} applied to $\mathbf{B}, $ converges in finite time and provides the following set that respects the conditions of Lemma \ref{polytopefirstdirection} 
$$\left\{\begin{pmatrix}
0 \\
0  \\
 \alpha^2 - 1 \\
\alpha^2 - 1 \\
0 \\
0   
\end{pmatrix},\
\begin{pmatrix}
0  \\
\alpha^2 - \alpha \\
0 \\
\alpha^2 - \alpha \\
\alpha^2 - \alpha \\
0  
\end{pmatrix},
\begin{pmatrix}
\alpha - 1 \\
\alpha - 1 \\
0 \\
\alpha - 1 \\
0 \\
\alpha - 1  
\end{pmatrix},
\begin{pmatrix}
- \alpha^2 + \alpha + 1  \\
- \alpha^2 + \alpha + 1 \\
0 \\
0  \\
0  \\
0 
\end{pmatrix}\right\}\,.$$
 We deduce our upper-bound.

Let us now show that this value of $\alpha$ is sharp even for paths.
Let $P_0= \basetree$ and for all $n$, $P_{n+1}= J(\basetree, P_n)$, then clearly $P_n$ is the path of order $n$ rooted at one of its end.
Let $\widetilde{P_0}=\mathbf{V_0}$ and for all $n$, $\widetilde{P_{n+1}}= \mathbf{B}(\mathbf{V_0}, \widetilde{P_n})$, then
$\mathbf{F}\cdot\widetilde{P_n}$ is the number of minimal perfect dominating set of the path of of order $n$.
Let $M$ be the matrix given by
$$M=\begin{pmatrix}
    0&0&0&1&0&0\\
    1&0&1&0&1&1\\
    0&0&0&0&0&0\\
    0&1&0&0&0&0\\
    1&0&0&0&0&0\\
    0&0&0&0&1&0
   \end{pmatrix}$$ then for all vector $\vect{v}$, $\mathbf{B}(\mathbf{V_0}, \vect{v})= M\vect{v}$ and $\widetilde{P_n}=M^n\mathbf{V_0}$.
The largest eigenvalue of $M$ is $\alpha$ so  the number of minimal perfect dominating set of the path of of order $n$ grows in $\Theta(\alpha^n)$.   
\end{proof}

\subsection[r-matchings]{$r$-matchings}
For any $r\ge1$ an $r$-matching is a set of edges $M$ such that any two edges of $M$ are at distance at least $r$. Matchings are exactly $1$-matchings and induced matchings are exactly induced matchings.  In other words, $S$ is an induced matching of a graph $(V,E)$ if any vertex from $S$ has exactly one neighbor in $S$. The authors of \cite{inducedmatching} initiated a study of the number maximal number of $r$-matchings in trees. The case of matching ($1$-matching) was already solved and it is known that the number of matching of a tree of order $n$ is maximized by the path of order $n$ \cite{matchings}.
In \cite{inducedmatching}, they showed that the number of $2$-matching  over trees is also maximized over paths. They asked whether this property also holds for $r$-matching with other values of $r$ and showed
that this is not the case for $r\notin\{3, 4, 5, 7, 9\}$ and left the remaining cases as an open question.
Using our technique it is rather simple to solve all the remaining cases. We show that for $r\in\{4,5,7,9\}$ the number of $r$-matchings over trees is not maximized over paths. 
In the case of $3$-matchings, we show that, up to a multiplicative constant, the number of $3$-matchings of the path of order $n$ is as big as the number of $3$-matchings of any tree of order $n$. A careful study of the vertices of the polytope would probably  allow to remove the multiplicative constant, but we leave this as an open question.

A set $S$ induces a matching of a graph if and only if every vertex from $S$ has exactly one neighbor in $S$. Moreover, the fact that two nodes are at least at distance $k$, for any fixed integer $k$, is easily expressible in $MSO_1$. Thus for any $r\ge2$, $r$-matchings of graphs are definable in $MSO_1$.
Matchings over graphs are more naturally expressed in $MSO_2$, but in trees $MSO_1$ has the same expressivity as $MSO_2$. Thus our technique is well suited to attack this question.
The cases of $1$-matchings and $2$-matchings were already treated, but we could easily provide an independent proof of these results in half a page (it goes exactly as the proof of Proposition \ref{minperfectdomsets}).

\begin{proposition}\label{3matchings}
 Let $\alpha\approx1.3802$ be the real root of $x^4-x^3-1$ between 1 and 2. 
 Then there exists a constant $C$ such that the number of $3$-matchings in a tree of order $n$ is at most $C\alpha^n$.
 Moreover, this value of $\alpha$ is sharp even for paths.
\end{proposition}
\begin{proof}
It was showed in \cite{inducedmatching} that $\alpha$ is the growth rate of the number of $3$-matchings over paths, so we only have to show that the bound holds for all trees.

For $3$-matchings, the initial vector, final vector and the bilinear map are respectively given by
$$\mathbf{V_0}=
\begin{pmatrix}
1\\0\\0\\1
\end{pmatrix},\,
\mathbf{F}=
\begin{pmatrix}
1\\  1\\  1\\  0
\end{pmatrix}\text{ and for all } u,v\in\mathbb{R}^4,\,
\mathbf{B}(x,y)=
\begin{pmatrix}
u_1v_1+ u_1v_2\\
u_1v_3+ u_2v_1+ u_2v_2\\
u_3v_1+ u_4v_4\\
u_4v_1\\
\end{pmatrix}\,. $$
 Algorithm \ref{algofindpolytope} converges in finite time and provides a set that respects the conditions of Lemma \ref{polytopefirstdirection} and we get our upper-bound. We provide this set in Annex \ref{m3sets}.
\end{proof}

\begin{proposition}\label{4matchings}
 Let $\alpha=13^\frac{1}{9}\approx 1.329754$. 
 Then there exists a constant $C$ such that the number of $4$-matchings in a tree of order $n$ is at most $C\alpha^n$.
 Moreover, this value of $\alpha$ is sharp.
\end{proposition}
\begin{proof}
 For $4$-matchings, the initial vector, final vector and the bilinear map are respectively given by
$$\mathbf{V_0}=
\begin{pmatrix}
1\\0\\0\\0\\1
\end{pmatrix},\,
\mathbf{F}=
\begin{pmatrix}
1\\  1\\  1\\1\\  0
\end{pmatrix}\text{ and for all } u,v\in\mathbb{R}^5,\,
\mathbf{B}(x,y)=
\begin{pmatrix}
u_1v_1+ u_1v_2\\
u_1v_3+ u_2v_1+ u_2v_2+ u_2v_3\\
u_1v_4+ u_3v_1+ u_3v_2\\
u_4v_1+ u_5v_5\\
u_5v_1
\end{pmatrix}\,. $$
The set given in Annex \ref{m4sets}, respects the conditions of Lemma \ref{polytopefirstdirection} and we get our upper-bound.\footnote{In this case Algorithm \ref{algofindpolytope} does not seem to converge in finite time. However, if we start with $\mathbf{V_0}$ and $\mathbf{V'}=(0, 1/13\alpha^8 + 1/6, 0, 0, 0)^T$ instead of only $\mathbf{V_0}$ the Algorithm converges in finite time. The set provided in Annex \ref{m4sets}, is in fact $\bigcup_{k\ge1}\mathbf{B}^k\left(\left\{\frac{\mathbf{V_0}}{\alpha}, \mathbf{V'}\right\}\right)$. Remark, that in $\mathbf{V'}$ the $1/6$ can be replaced by anything between $1/6$ and $1/100$ (that could be true for any positive real smaller than $1/6$).  This underlines the fact that the convex set that respects the conditions of Lemma \ref{polytopefirstdirection} is not necessarily unique.}

For our lower bound, let $T_k$ be the tree constructed from $k$ copies of $P_9$ and a vertex $v$ such that there is an edge from the central vertex of each $P_9$ to $v$. We give an illustration of $T_5$ in Fig \ref{tree4matchings}.
\begin{figure}[H]
\centering
\includegraphics[scale=0.5]{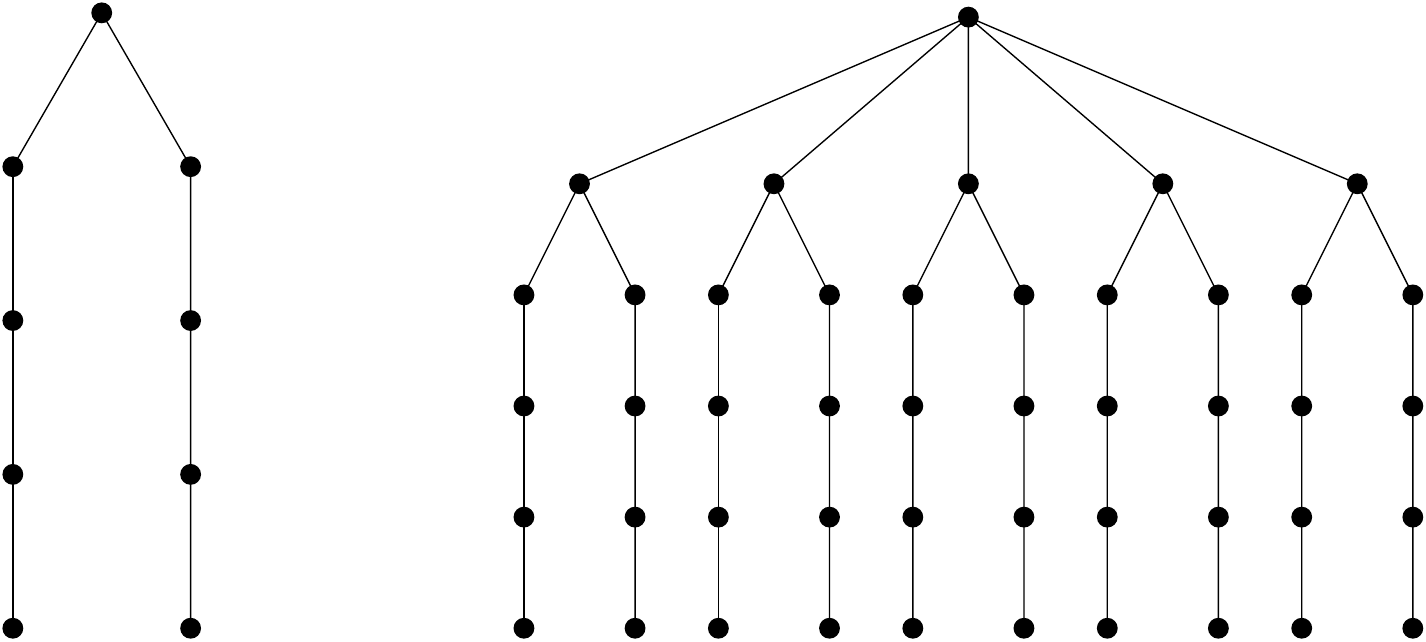}
\caption{On the left $P_9$ and on the right the tree $T_5$. \label{tree4matchings}}
\end{figure}

In each copy of the $P_9$ there are $13$ $4$-matchings that do not use the central vertex. If we chose one of these $4$-matchings for each copy then the union is a $4$ matching of the tree. Thus there are at least $13^k= 13^\frac{n-1}{9}$ $4$-matchings in $T_k$. 
Thus the maximal number of $4$-matchings of trees of order $n$ grows in $\Theta(\alpha^n)$.
\end{proof}
\begin{proposition}\label{5matchings}
 Let $\alpha= \frac{22}{17}\approx 1.29411$ and $\beta\approx 1.293211$ be the only positive real root of 
 \begin{align*}\scriptstyle
   x^{6\times45}
   -104625x^{5\times45}
   -14946778125000x^{3\times45}
   -28242953648100000000x^{45}
   -7230196133913600000000
 \end{align*}
 There exists a constant $C$ such that the number of $5$-matchings in a tree of order $n$ is at most $C\alpha^n$. There exists a family of trees that have $\theta(\beta^n)$  $5$ matchings.
\end{proposition}
\begin{proof}
 For $5$-matchings, the initial vector, final vector and the bilinear map are respectively given by
$$\mathbf{V_0}=
\begin{pmatrix}
1\\0\\0\\0\\0\\1
\end{pmatrix},\,
\mathbf{F}=
\begin{pmatrix}
1\\  1\\  1\\ 1\\ 1\\  0
\end{pmatrix}$$
$$\text{ and for all } u,v\in\mathbb{R}^6,\,
\mathbf{B}(x,y)=
\begin{pmatrix}
u_1v_1+ u_1v_2\\
u_1v_3+ u_2v_1+ u_2v_2+ u_2v_3\\
u_1v_4+ u_2v_4+ u_3v_1+ u_3v_2+ u_3v_3\\
u_1v_5+ u_4v_1+ u_4v_2\\
u_5v_1+ u_6v_6\\
u_6v_1
\end{pmatrix}\,. $$
The computation of $\bigcup_{k\ge1}\mathbf{B}^k(\frac{\mathbf{V_0}}{\alpha})$ by Algorithm \ref{algofindpolytope} converges in finitely many steps and we reach a finite set $X$
that respects the conditions of Lemma \ref{polytopefirstdirection}.\footnote{This set $X$ was computed in less than a minute on a laptop and contains $59$ vectors. However, the numerators and denominators of the rationals coordinates reach such high values that writing the set takes around $54000$ letters and it would take a few hundred pages to include this set in an annex. It seems more efficient to provide the C++ code that computes and verifies this set.}
This concludes the proof of the upper-bound.

For our lower bound, let $T$ be the rooted tree given in Figure \ref{tree5matchings} 
and let $T_k$ be the tree made of $k$ copies of $T$ connected by a path going through the roots of each copy. 
\begin{figure}[H]
\hspace{10px}
\includegraphics[scale=0.4]{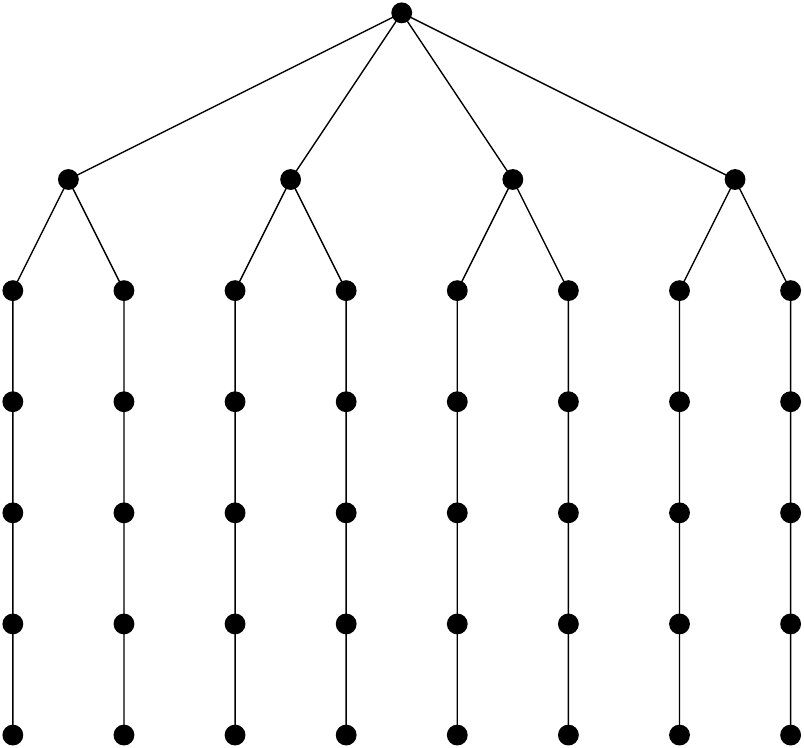}
\hfill
\includegraphics[scale=0.5]{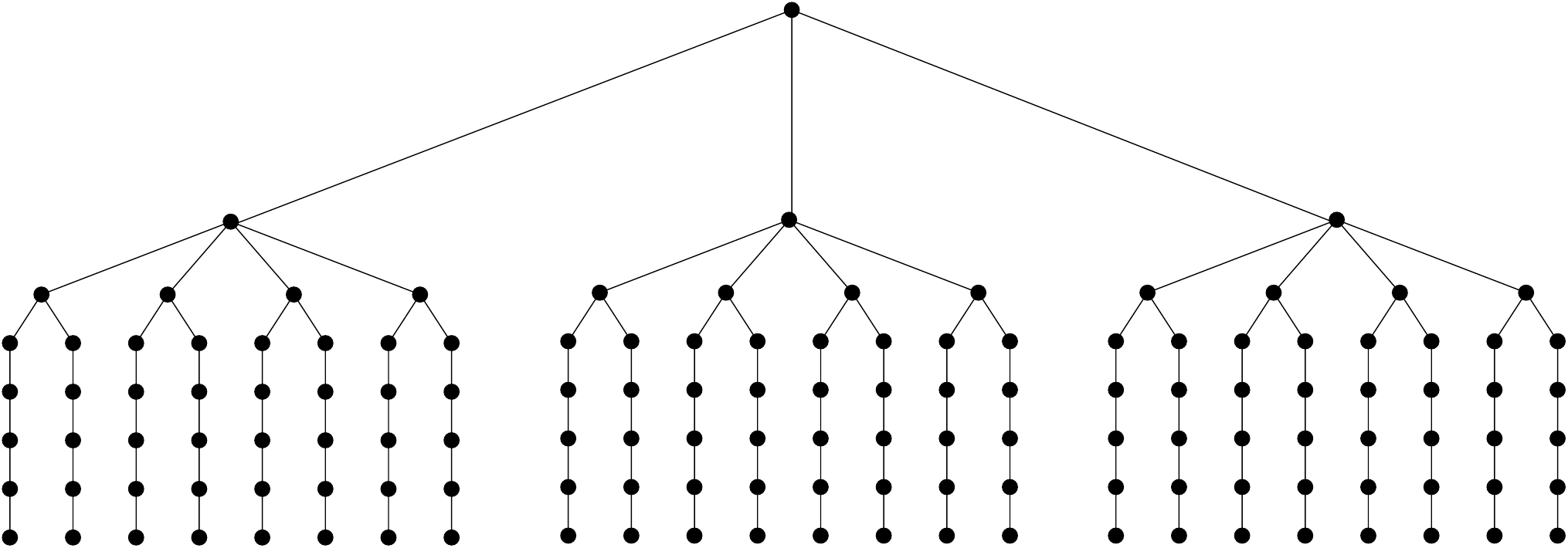}
\hspace{10px}
\caption{The tree $T$ on the left and $T_3$ on the right. \label{tree5matchings}}
\end{figure}
Let $\mathbf{V_{P_{11}}}$ be the vector corresponding to the path over $11$ vertices rooted in its central vertex, that is
$$\mathbf{V_{P_{11}}}=\mathbf{B}(\mathbf{B}(\mathbf{V_0},\mathbf{B}(\mathbf{V_0},\mathbf{B}(\mathbf{V_0},\mathbf{B}(\mathbf{V_0},\mathbf{B}(\mathbf{V_0},\mathbf{V_0})))))
,\mathbf{B}(\mathbf{V_0},\mathbf{B}(\mathbf{V_0},\mathbf{B}(\mathbf{V_0},\mathbf{B}(\mathbf{V_0},\mathbf{V_0})))))\,.$$
Then let $\mathbf{V}_T$ be the vector corresponding to $T$, that is
$$\mathbf{V}_T=\mathbf{B}(\mathbf{B}(\mathbf{B}(\mathbf{B}(\mathbf{V_0},\mathbf{V_{P_{11}}}),\mathbf{V_{P_{11}}}),\mathbf{V_{P_{11}}}),\mathbf{V_{P_{11}}})\,.$$
Let $M$ be the matrix such that for all $x\in\mathbb{R}^6$, 
$Mx= \mathbf{B}(\mathbf{V}_T,x)$ then $FM^{k-1}\mathbf{V_{P_{11}}}$ is the number of $5$-matchings of $T_k$.
Explicit computation of $M$ gives
$$ M= \begin{pmatrix}
6561&   6561&   0&   0&   0&   0\\ 
44064&   44064&   50625&   0&   0&   0\\    
54000&   54000&   54000&   50625&   0&   0\\    
5832&   5832&   0&   0&   6561&   0\\    
256&   0&   0&   0&   0&   256\\    
256&   0&   0&   0&   0&   0
\end{pmatrix}\,.$$
and we easily compute that $\beta^{45}$ is the largest eigenvalue of $M$. This implies that the number of $5$-matchings of $T_k$ is in $\theta(\beta^{45k})$. Since $T$ contains $45$ vertices this implies that the number of $5$-matchings of $T_k$ is in $\theta(\beta^{n})$ where $n$ is the number of vertices of $T_k$.
\end{proof}
This is the first case where we are not able to provide the exact growth rate. However, we are still able to provide a really good bound and it implies that the maximal number of $5$-matching is not reached on a path (since the number of $5$ matchings of a path is $\approx1.2852$ \cite{inducedmatching}). We believe that neither the upper-bound nor the lower-bound is sharp. 

In fact, the construction of Proposition \ref{4matchings} can be adapted to a construction with more $7$-matching (resp. $9$ matchings) than the path.
\begin{proposition}\label{79matchings}
Let $r\ge3$ be an even integer.Then the number of $r$-matching in trees of order $n$ is at least $\Omega(15^{\frac{n}{r+6}})$.
\end{proposition}
\begin{proof}
Let $T_{r,k}$ be the tree built from $k$ copies of $P_{r+6}$ all connected to the root $v$.

The number of $r$-matchings of $P_{r+6}$ that do not include one of the $r-2$ central nodes is $15$. Indeed, there are 4 free vertices  on each side and thus, 4 possible way to select at most one edge on each side (including the one where no edge is selected), the only forbidden configuration is the one where the most centrals non forbidden edges are selected on each side (since this two edges are at distance $r-1$). This give us $4^2-1=15$ such $r$-matchings.

Now if we select one of these $r$-matchings for every copy of $P_{r+6}$, we get an $r$-matching of $T_{r,k}$ since two edges selected in two different $P_{r+6}$, are at distance at least $2\frac{r-1}{2}+2= r+1$. Thus $T_{r,k}$ has at least $15^k$ $r$-matchings.
Since $T_{r,k}$ has $(r+6)k+1$ nodes this concludes our proof.
\end{proof}
The lower bound from this Proposition is strong enough to verify that the number of $7$-matching (resp. $9$-matchings) over trees are not maximized by the path (the number $7$-matching (resp. $9$-matchings) of the path can be found in \cite{inducedmatching}).

Remark that these constructions do not maximize the number of $r$-matchings. We could obtain slightly better lower bounds by simply connecting the $P_{r+6}$ by a path going through the central vertices, but constructions similar to the one from Proposition \ref{5matchings} provide even better lower bounds. However, the computation become really complicated and this is not the goal of this article to discuss in detail the case of $r$-matchings. 
We could also use our method to obtain upper bounds on the number of $7$-matchings  and $9$- matchings in trees.

\subsection{Maximal matchings}
A matching is maximal if it is not strictly contained in another matching.
The authors of \cite{GORSKA20071367} showed that the maximal number of maximal matchings of a
tree of order $n$ grows in  $\Omega(1.39097^n)$ and  $O(1.395337^n)$ (the precise values of the constants involved can be found in \cite{GORSKA20071367}). More recently, Heuberger and Wagner showed that the maximal number of maximum matchings (matchings of maximal size) in trees grows in $\Theta\left(\left(\frac{11+\sqrt85}{2}\right)^\frac{n}{7}\right)$\cite{HEUBERGER20112512}. Since every maximum matching is maximal this improves the lower bound on the maximal number of maximal matchings from  $\Omega(1.39097^n)$ to $\Omega(1.39166^n)$.
S. Wagner conjectured (personal communication) that the bound for maximum matchnings also holds for maximal matchings and we are able to verify this conjecture.

\begin{proposition}\label{simplematchings}
 Let $\alpha =\left(\left(\frac{11+\sqrt85}{2}\right)^\frac{n}{7}\right)\approx  1.391664$ and $C=-1/3\alpha^{10} + 11/3\alpha^3$. The number of maximal matchings of a tree of order $n$ is less than $\alpha^n$. 
 Moreover, this value of $\alpha$ is sharp.
\end{proposition}
\begin{proof}
In general matchings are not expressible in $MSO_1$, however in trees $MSO_1$ and $MSO_2$ are as expressive since once can always root the tree and represent an edge by the corresponding child (an edge is always between a node an a child of this node). Thus, our approach can be used for this case as well.
For matchings, the initial vector, final vector and the bilinear map are respectively given by
$$\mathbf{V_0}=
\begin{pmatrix}
0\\  1\\  0\\  1  
\end{pmatrix},\,
\mathbf{F}=
\begin{pmatrix}
1\\  0\\  0\\ 0
\end{pmatrix}$$
$$\text{ and for all } u,v\in\mathbb{R}^4,\,
\mathbf{B}(x,y)=
\begin{pmatrix}
  u_{1}v_{1}+ u_{1}v_{2}+ u_{2}v_{4}+ u_{3}v_{4}\\
  u_{2}v_{1}\\
  u_{2}v_{2}+ u_{3}v_{1}+ u_{3}v_{2}\\
  u_{4}v_{1}+ u_{4}v_{2}
\end{pmatrix}\,. $$
The set given in Annex \ref{maxmatch} respects the conditions of Lemma \ref{polytopefirstdirection}. 
This concludes the proof of the upper-bound.

The fact that the value of $\alpha$ is sharp was established by the previously mentionned result of 
Heuberger and Wagner \cite{HEUBERGER20112512}.
\end{proof}
The set provided in Annex \ref{maxmatch} corresponds to $\bigcup_{k\ge1}\mathbf{B}^k(\frac{\mathbf{V_0}}{\alpha})$ but does not seem to be reached by Algorithm \ref{algofindpolytope} in finitely many steps. 
However, the set $\bigcup_{k\ge1}\mathbf{B}^k(\{\frac{\mathbf{V_0}}{\alpha},v\})$ where $v$ is a well chosen vector
can be computed in finitely many step by our algorithm and is the set provided in  Annex \ref{maxmatch}.

The vertice $v$ is given by
$v=
\,^T(5/153\alpha^{11} - 19/153\alpha^4,  0,   - 4/765\alpha^{11} + 107/765\alpha^4,   - 4/765\alpha^{11} + 107/765\alpha^4 )$ and is a well chosen eigenvector of the matrix
$$\begin{pmatrix}
8&   3&   0&   5\\  
0&   0&   0&   0\\  
3&   0&   0&   3\\   
3&   0&   0&   3
  \end{pmatrix}\,.$$
This matrix is the linear operation that correspond to applying the operation ovezr rooted trees depicted in Figure
\ref{figsimplematchings}. The precise scaling of this eigenvector was chosen by studying the Jordan decomposition of the matrix as well as the vectors that correspond to extremal vertices obtained when trying to compute $\bigcup_{k\ge1}\mathbf{B}^k(\frac{\mathbf{V_0}}{\alpha})$. Remark that the eigenvalue associated to this eigenvector is $\alpha$ which tells us that iterating the construction from Figure \ref{figsimplematchings} gives a sequence of trees that assymptoticaly have the maximum number of maximal matchings.
\begin{figure}
\centering
\includegraphics[scale = 3]{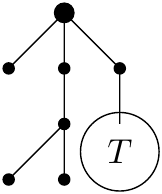}
\caption{An illustration of the construction of that assymptoticaly maximizes the number of maximal matchings \label{figsimplematchings}}
\end{figure}

\subsection{Maximal induced matchings}
An induced matching is maximal if it is not strictly contained in another induced matching.
In the case of maximal induced matching the computations become more complicated than in the case of induced matchings and instead of giving the exact growth rate we are only able to provide really good approximations.

The initial vector, final vector and the bilinear map corresponding to these sets are respectively given by
$$\mathbf{V_0}=
\begin{pmatrix}
0\\  0\\  1\\  1\\  0  
\end{pmatrix},\,
\mathbf{F}=
\begin{pmatrix}
1\\1\\1\\0\\0
\end{pmatrix}
\text{ and for all } u,v\in\mathbb{R}^5,\,
\mathbf{B}(x,y)=
\begin{pmatrix}
u_1v_2+ u_1v_3+ u_1v_5+ u_4v_4\\
u_2v_1+ u_2v_2+ u_2v_3+ u_3v_1+ u_5v_1\\
u_3v_2\\
u_4v_2+ u_4v_3+ u_4v_5\\
u_3v_3+ u_5v_2+ u_5v_3\\
\end{pmatrix}\,. $$

\begin{proposition}\label{Max_induced_matching_prop}
 Let $\alpha =\frac{4254960628685}{3195429966304}\approx1.331576$ and $\beta\approx1.331576$ be the real root of $108 - 135 x^8 + 132 x^{16} - 33 x^{24} + 12 x^{32} - x^{40}$ between 1 and 2.
 Then the number of maximal induced matchings in a tree of order $n$ is less than $\alpha^n$.
 Moreover, there exists a constant $C\in \mathbb{R}_{>0}$ and a sequence of trees with more than $C\beta^n$ maximal induced matchings.
\end{proposition}
\begin{proof}
The computation of $\bigcup_{k\ge1}\mathbf{B}^k(\frac{\mathbf{V_0}}{\alpha})$ by Algorithm \ref{algofindpolytope} converges in finitely many steps and we reach a finite set $X$
that respects the conditions of Lemma \ref{polytopefirstdirection}.\footnote{This set $X$ was computed in under 7 minutes on a laptop and contains $80$ vectors. However, the numerators and denominators of the rationals coordinates reach such high values that writing the set takes around $220000$ letters and it would take a few hundred pages to include this set in an annex. It seems more efficient to provide the C++ code that computes and verifies that Lemma \ref{polytopefirstdirection} applies to this set.}
This concludes the proof of the upper-bound.

Now, let $M$ be the matrix such that for all $x\in\mathbb{R}^5$,
$$Mx= \mathbf{B}(\mathbf{B}(\mathbf{V_0},\mathbf{B}(\mathbf{V_0},\mathbf{B}(\mathbf{V_0},\mathbf{B}(\mathbf{B}(\mathbf{V_0},\mathbf{V_0}),\mathbf{V_0})))),\mathbf{B}(\mathbf{V_0},\mathbf{B}(\mathbf{V_0},x)))\,.$$
One easily verifies, that for all $x\in\mathbb{R}^5$,
$$\mathbf{B}(\mathbf{V_0},x)=
 \begin{pmatrix}
0&   0&   0&   1&   0\\
1&   0&   0&   0&   0\\
0&   1&   0&   0&   0\\ 
0&   1&   1&   0&   1\\   
0&   0&   1&   0&   0  
 \end{pmatrix} x\,.$$
 Similarly,
$$\mathbf{B}(\mathbf{V_0},\mathbf{B}(\mathbf{V_0},\mathbf{B}(\mathbf{V_0},\mathbf{B}(\mathbf{B}(\mathbf{V_0},\mathbf{V_0}),\mathbf{V_0})))=(2,  1,  1,  3, 2)^{T}$$ and thus for every $x\in\mathbb{R}^5$,
 $$\mathbf{B}(\mathbf{B}(\mathbf{V_0},\mathbf{B}(\mathbf{V_0},\mathbf{B}(\mathbf{V_0},\mathbf{B}(\mathbf{B}(\mathbf{V_0},\mathbf{V_0}),\mathbf{V_0})))),x)=
 \begin{pmatrix}
0&   2&   2&   3&   2\\
4&   1&   1&   0&   0\\
0&   1&   0&   0&   0\\ 
0&   3&   3&   0&   3\\   
0&   2&   3&   0&   0  
 \end{pmatrix} x\,.$$
 We deduce that $$M= \begin{pmatrix}
0&   2&   2&   3&   2\\
4&   1&   1&   0&   0\\
0&   1&   0&   0&   0\\ 
0&   3&   3&   0&   3\\   
0&   2&   3&   0&   0  
 \end{pmatrix}
  \begin{pmatrix}
0&   0&   0&   1&   0\\
1&   0&   0&   0&   0\\
0&   1&   0&   0&   0\\ 
0&   1&   1&   0&   1\\   
0&   0&   1&   0&   0  
 \end{pmatrix}^2=
  \begin{pmatrix}
5&   5&   3&   2&   0\\ 
1&   4&   4&   1&   4\\
0&   0&   0&   1&   0\\  
3&   3&   0&   3&   0\\
3&   0&   0&   2&   0
  \end{pmatrix}\,. $$
The characteristic polynomial of $M$ is $-x^5 + 12 x^4 - 33 x^3 + 132 x^2 - 135 x + 108$. By definition, $\beta^8$ is  a root of this polynomial and one easily verifies that this is one of the roots of largest absolute value. Since $M$ is primitive, we deduce that $|FM^n\mathbf{V_0}|$ grows in $\theta(\beta^{8n})$ as $n$ goes to infinity. Since, for all $n$, $M^nV_0\in\mathbf{B}^{8n+1}(\mathbf{V_0})$, we deduce that the growth rate of $\mathbf{B}$ is at least $\beta$.
\end{proof}
Let us first discuss the bounds given in this proposition. Remark, that the upper and lower bounds are both exact bounds and not floating-point approximations. Moreover, 
$\alpha-\beta<7\times 10^{-26}$, thus even if the upper and lower bounds do not match we can still advertise these bounds as really good. In fact, it seems that one easily gets arbitrarily good bounds here since in this case the algorithm seems to converge with any rational upper-bound and the lower-bound seems to be optimal. Going from a precision of $2\times 10^{-18}$ (with $\frac{933335285}{700924826}$) to a precision of $7\times 10^{-26}$ only increased the computation time from 5 minutes to 7 minutes. That might be the case that there is actually a finite polytope $X$ that allow us to close the gap between the upper and the lower-bound but we were not able to find it (if it exists we expect that one of the eigenvectors of $M$ corresponding to $\beta$ should be a vertex of the polytope, but it is not even clear which scaling should be applied).

We should also provide more intuition about the construction that provided the lower-bound. 
The operation corresponding to our construction that takes a tree $T$ and produces  a new tree is depicted in Figure \ref{newtreeindmax}. We give an instance of this construction applied iteratively 6 times to the path on 2 vertices in Figure \ref{exampletreeindmax}.
The construction was found by inspection of the vertices of $ \operatorname{conv}\left(\bigcup_{n\le k\ge1}\mathbf{B}^k(\frac{\mathbf{V_0}}{\beta'})\right)$ with $\beta'$ a real number close to but smaller than $\beta$ and $n$ a relatively large integer. Trees similar to the one from Figure \ref{exampletreeindmax} started to appear clearly when we reached vertices of the convex hull corresponding to trees with $200$ vertices (it took $\approx$ 10 minutes on a laptop).

\begin{figure}
\centering
\begin{minipage}[c]{0.45\textwidth}
\centering

\includegraphics[scale = 3.3]{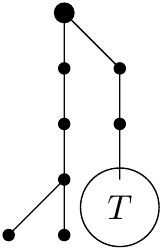}
\caption{An illustration of the construction of the lower bound of Proposition\ref{Max_induced_matching_prop} \label{newtreeindmax}}
    
\end{minipage}
\hspace{0.05\textwidth}
\begin{minipage}[c]{0.45\textwidth}
\centering
\includegraphics[scale = 0.15]{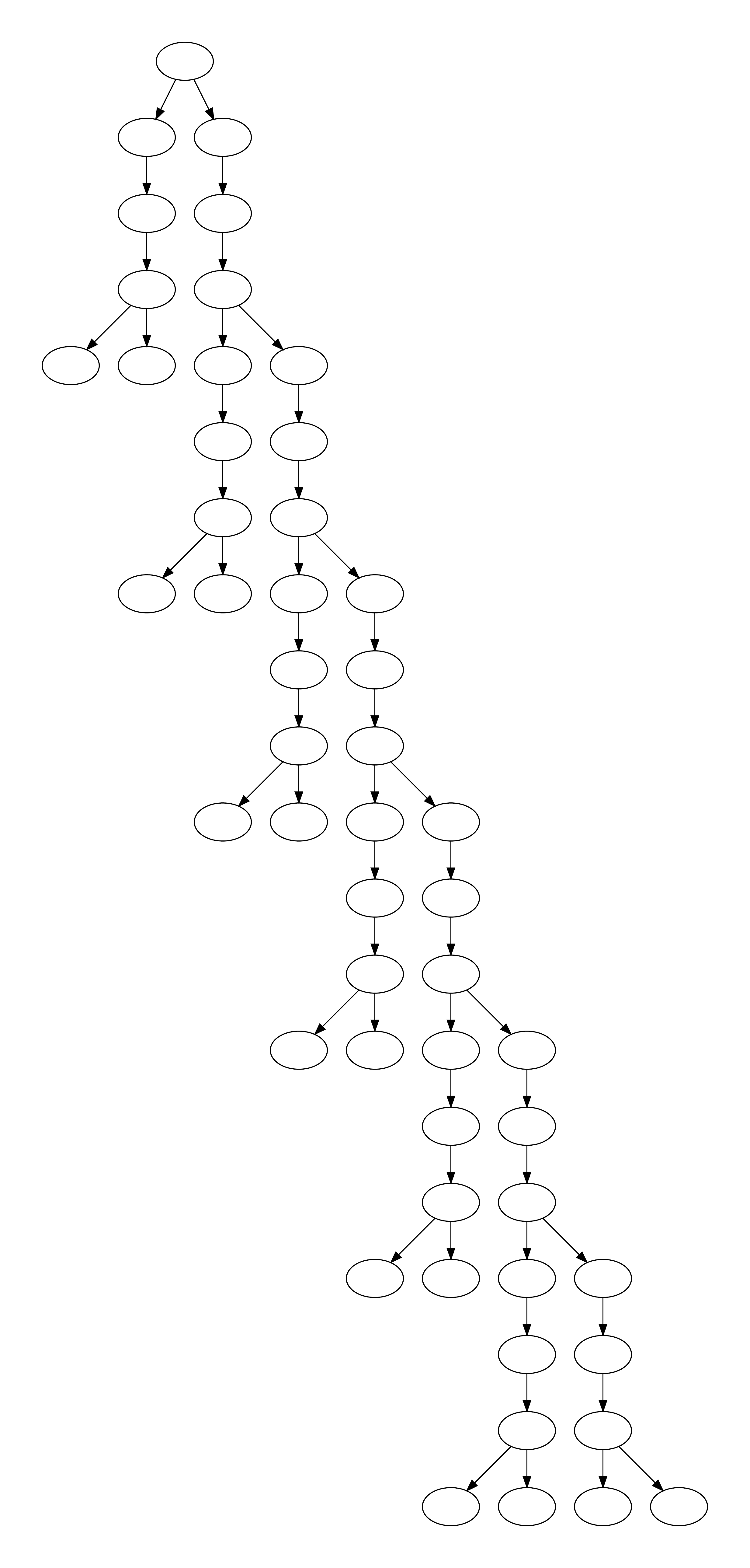}
\caption{An instance of the construction of the lower bound of Proposition\ref{Max_induced_matching_prop} \label{exampletreeindmax}}
\end{minipage}
\end{figure}

\subsection{Maximal irredundant sets}
The closed neighborhood of a vertex $v$, denoted by $N[v]$ is the set of neighbors of $v$ and $v$.
For any graph $(V,E)$, any set $S\subseteq S$ and vertex $v\in S$, we say that a vertex $u$ is a private neighbor of $v$ if $N[u]\cap S= \{v\}$. 
In other words, $u$ is a private neighbor of $v$ if $v$ is the only element from $S$ in the closed neighborhood of $u$ is $v$.
We say that a set $S$ is an irredundant set if every vertex from $S$ has a private neighbor.
An irredundant set is maximal if it is not strictly contained in another irredundant set.
Maximal irredundant sets received a lot of attention through the literature in particular because of their similarities with minimal dominating sets.
Indeed, every set is a minimal dominating set if and only if it is dominating and irredundant, and this implies that every minimal dominating set is a maximal irredundant set.

One easily verifies that the number of irredundant sets of a star on $n$ vertices is $2^{n-1}$. Bounding the number of maximal irredundant sets is much harder than the other studied sets and we are only able to provide good approximations.

\begin{proposition}\label{Max_irr_sets}
 Let $\alpha =\frac{14}{9}\approx1.555556$ and $\beta=48^{\frac{1}{9}}\approx 1.53746$.
 Then the number of maximal irredundant sets in a tree of order $n$ is less than $\alpha^n$.
 Moreover, there exists a constant $C\in \mathbb{R}_{>0}$ and a sequence of trees with more than $C\beta^n$ maximal irredundant sets.
\end{proposition}
\begin{proof}
We provide the bilinear system associated with maximal irredundant sets over trees in Annex \ref{annex_irr_sets}.
The computation of $\bigcup_{k\ge1}\mathbf{B}^k(\frac{\mathbf{V_0}}{\alpha})$ by Algorithm \ref{algofindpolytope} converges in finitely many steps and we reach a finite set $X$
that respects the conditions of Lemma \ref{polytopefirstdirection}.\footnote{This set $X$ was computed in approximatively 4 hours on a laptop and contains $393$ vectors. It is once again too large to inlude in an Annex ($277000$ letters), but we provide the C++ code that computes and verifies this set.}
This concludes the proof of the upper-bound.

Let $T$ be the tree over $9$ vertices built from a central vertex with $4$ pendant paths over two vertices (illustrated in Figure \ref{constrmirs}). Then let $T_k$ be the tree built from $k$ copies of $T$ and two vertices $s_1$ and $s_2$ such that the central vertex of each copy of $T$ shares an edge with $S_1$ and $s_1$ and $s_2$ share an edge (illustrated in Figure \ref{constrmirs}).
\begin{figure}[H]
\includegraphics[scale = 1.2]{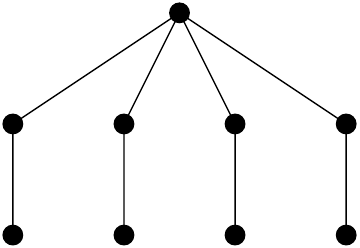}
\hfill
\includegraphics[scale = 0.8]{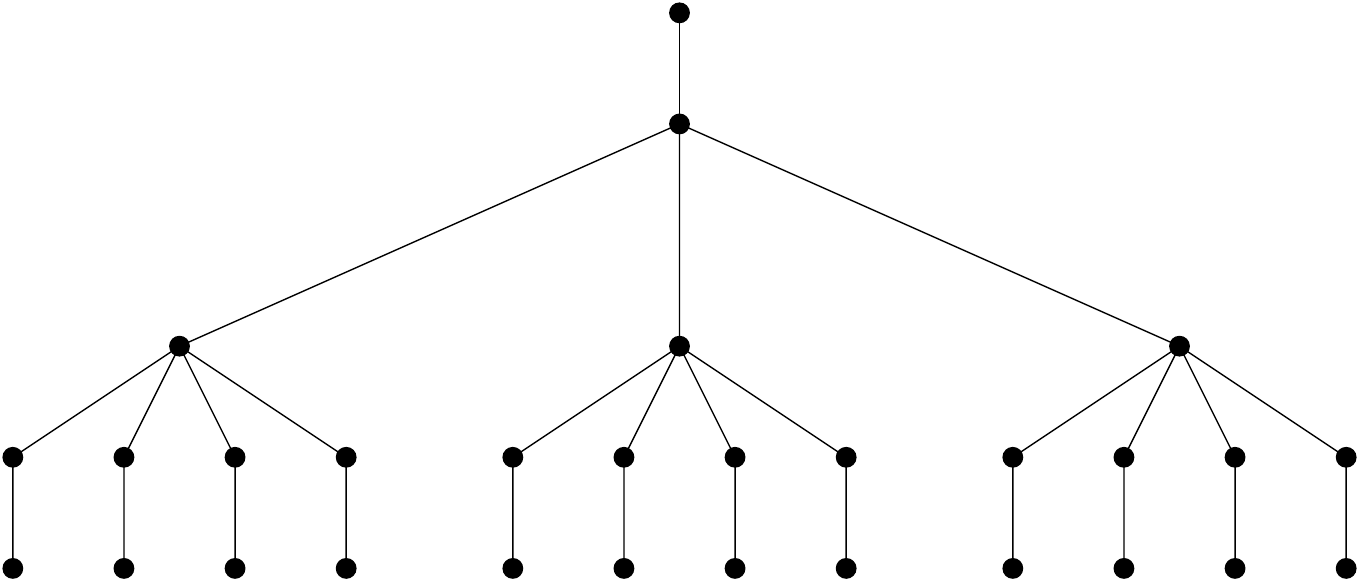}
\caption{The tree $T$ on the left and $T_3$ on the right.\label{constrmirs}}
\end{figure}
Let $S$ be a maximal irredundant set of $T_k$ such that $s_1\in S$.
First remark, that $s_2\not\in S$ and $S_2$ is a private neighbor of $S_1$ so $s_1$ does not need another private neighbor. 

Let us count the number of ways of choosing $S$ in a given copy of $T$. Recall, that $T$ has a central vertex and four $P_2$. We say that one of the $P_2$ is \emph{occupied} if one of its vertex is in $S$ (remark, that at most one of the two vertices is in $S$). If a $P_2$ is occupied then the vertex in $S$ has a private neighbor, but none of the two vertices can be a private neighbor of the central vertex. If two $P_2$ are not occupied, then $S$ is not maximal. Thus either all the $P_2$ are occupied and the central vertex is not in $S$ or exactly one of the $P_2$ is not occupied. This gives $2^4+4\times2^3=48$ ways of choosing $S$ in each copy of $T$. 

We deduce that there are at least $48^k$ maximal irredundant sets in $T_k$.
Since $T_k$ has $9k+2$ vertices this conclude the proof of our lower-bound.
\end{proof}
Although $0.02$ is a small gap between the lower and the upper bounds, this is the worst bound in this article. This is not surprising, since it is from far the bilinear system of higher dimension. We are indeed, working in $\mathbb{R}^{20}$, while the previous highest dimension system was in $\mathbb{R}^6$. This makes each step of computation much slower, but also increases the number of steps since one might expect the best convex polytope to have more vertices (this statement is to be understood as a vague intuition and not a hard fact, since this is probably the case that some systems in low dimension have no convex polytopes with finitely many vertices). This illustrates well the computational limits of our approach, in particular if one wishes to obtain sharp bounds. However, they are probably many ways to improve the algorithmic complexity of this approach.

\section{Generalization and conclusion}
The  main idea presented here was to use logic to deduce automatically from the definition of a kind of set the bilinear map associated to a tree automaton recognizing accepted sets and to use basic linear algebra to deduce the growth rate of this family of sets.
We were able to obtain good bounds for all the families of sets that we studied which is rather promising (they all appear here except for minimal dominating sets and induced matching that were already solved). 
Before generalizing this approach it would be interesting to find other interesting applications on trees.

There are many obvious generalizations of our approach.
First, using standard techniques, we can replace trees, by graphs of tree-width or clique-width (or any other similar parameters) at most $p$ for some fixed $p$. In this cases, the alphabets of the terms corresponding to a graph have more than one letter of arity $2$ and one has to project the tree-automaton on each of this letter to obtain one bilinear map for each such letter. The growth rate of a set of bilinear maps can be approximated from above using exactly the same technique.
One can also replace $MSO_1$ by $MSO_2$ even for graphs of bounded tree-width. In $MSO_2$ quantification over set of edges are also allowed and the only proof that need to be slightly modified is the proof of Theorem \ref{MSO_1toWS2S} (in fact, over trees $MSO_1$ and $MSO_2$ are equivalent since one can always root the tree and represent an edge by the node whose unique path to the root takes this edge first). This would for instance allow to bound the number of perfect matchings in trees. In fact, second order monadic logic can even be replaces by counting second order monadic logic (one is allowed to test if the cardinally of a set is $\equiv a\mod b$ where $a$ and $b$ are any fixed integers). We could also handle tuples of sets instead of sets and we could, for instance, bound the number of proper $3$-coloring of a tree.
It is not clear that any  of these cases have any interesting applications and that the computation can still be done by a computer.

We can also restrict this approach to smaller classes. One can replace trees by paths or graphs of path-width at most $p$ for some fixed $p$. Then  tree-automata are replaced by automata and the bilinear maps by matrices. In this case, the computation of the growth rate is reduced to the computation of the joint spectral radius of a set of matrices. The joint spectral radius has been widely studied and in particular there are known algorithms that provides arbitrarily good approximations. In the case of paths there is a single matrix and we can always obtain an explicit formula for the tight bound.

\subsection{More on the computation of the growth rate of a bilinear system}
The idea of finding a convex polytope invariant by our bilinear operator was already used by Rote wfor the computation of the growth rate of the maximal number of minimal dominating sets of trees \cite{r-mnmds-19,arxivRote}. This idea is in fact really close to \emph{the polytope norm method} used for the computation of the joint spectral radius (see for instance \cite{CICONE2010796} for a definition and an application of this method).

In his paper Rote already raised question regarding the computability of the growth rate of a bilinear system. In this paper we were able to show that this quantity is approximable from above. However, we suspect that this quantity is in fact computable. Experimentally it is rather easy to detect that $\mathbf{B}^k\left(
\mathbf{V_0}\right)$ diverges. Being able to do so would allow to decide whether the growth rate is greater than $1$ which would lead to an approximation algorithm.

One of the question raised by  Rote is ``Is it sufficient to consider bodies that are polytopes?''. 
If one is interested in approximation of the growth rate Corollary \ref{polytopesaregoodenough} implies that it is even sufficient to take polytopes with rationals coordinates. 
On the other-hand if the growth rate is exactly one, the set $\mathbf{B}^k\left(
\mathbf{V_0}\right)$ is not necessarily bounded. As for instance the number of induced paths of a tree is quadratic in the size of the tree (choose the two ends of the path) and induced paths are definable in $MSO_1$.
Thus if one desires exact computations of the growth rate then there might not even be any suitable convex body.

\section*{Acknowledgement}
The author would like to thank Stephan Wagner who suggested us a conjecture on the maximal number of maximal matchings after reading a previous version. His conjecture helped us find the sharp upper-bound of Lemma \ref{simplematchings} instead of a previous bound that was $10^{-6}$ away from the sharp bound.

\bibliographystyle{plain}
\bibliography{biblio}

\appendix
\section{Annex: C++ code}\label{C++desc}
We give a short description of the code used to generate the bilinear maps and find the sets $X$ by applying Algorithm \ref{algofindpolytope}.
\footnote{The files can be downloaded from the ancillary files at the arXiv repository.}
The main scripts reads the user inputs, calls Mona in order to obtain a tree automaton and calls our C++ program to compute the growth of the bilinear system. 
In order to use this one needs bash to call the main script, g++ and gmp to compile the code (gmp is a library that among other things implements rational numbers) and Mona. The only requirement that is not installed in a simple command line on a Linux is Mona, but it is still rather simple to install.\footnote{Mona can be downloaded from here \url{https://www.brics.dk/mona/download.html}}

There are fives useful files for our script.
The script himself is in  \textbf{Growth\_in\_trees.sh}. The script starts by checking that g++ and Mona are properly installed and compiling the C++ files. Then it asks the user to either chose from a list or provide an MSO formula. It calls Mona on the chosen formula and reads the output. It cleans the output of Mona from everything else than the description of the tree automaton and calls the C++ program. The list of pre-configured inputs is stored 
in the file \textbf{mso\_formulas} in a text format that can easily be modified. 
The program (generated from the file \textbf{read\_Mona\_Output\_Tree.cpp}) then asks the user to provide a possible value for the growth rate and applies Algorithm \ref{algofindpolytope} to the bilinear system. This program runs until it success, thus it might be necessary to kill.

Two other files are dedicated to the computation of the set $X$:
\begin{description}
 \item[\textbf{algebraic.hpp:}] This file implements the class \emph{Number} that allows us to do exact computation on algebraic numbers. 
 This relies on the bijection between the smallest algebraic extension of $\mathbb{Q}$ that contains $\alpha$ and $\mathbb{Q}[X]/P(X)$ where 
 $\alpha$ is an algebraic number of minimal polynomial $P(X)$. Interval of rationals (that can be computed with arbitrary precision) are used to solve inequalities. 
 The only non-trivial operation is the division, but this can be done by using the extended Euclidean algorithm to compute Bézout coefficients and obtain the inverse of a polynomial in $\mathbb{Q}[X]/P(X)$.
 \item[\textbf{X\_from\_operators.hpp:}] This file implements Algorithm \ref{algofindpolytope}. We use a simple implementation of the simplex algorithm to find the set $\operatorname{conv}_\le(S)$.
\end{description}
For more details on the implementation, one should look at the code. 

Let us now shortly describe the syntax of MSO formulas expected by our program.
The only free variable of the formula should be the second order variable $S$ that defines the desired family of sets.
The syntax of atomic formulas is as follows $$\phi:= edge(x,y)|x = y | x \sim= y|T = R|T \sim= R|x\ in\ T | x\ notin\ T$$ where $x$ and $y$ are first order variables and $T$ and $R$ are second order variables.
Formulas are build from atomic formulas using disjunctions ($\phi_1|\phi_2$), conjunctions ($\phi_1\&\phi_2$), implications ($\phi_1=>\phi_2$) and negations ($\sim \phi_2$). Each formula should end with a semicolon (;) and giving multiple formulas is the same as giving the conjunctions of these formulas. Quantification is slightly more complicated since one has to take care that the manipulated variables are valid that is:
\begin{itemize}
 \item first order existential quantification of $y$ is written \verb+ex1 y: validVertex(y) & +$(\phi)$ 

 \item first order universal quantification of $y$ is written \verb+all1 y: validVertex(y) => +$(\phi)$ 

 \item second order universal quantification of $y$ is written \verb+all2 y: validSet(y) & +$(\phi)$ 

 \item second order universal quantification of $y$ is written \verb+all2 y: validSet(y) => +$(\phi)$ 
\end{itemize}
For instance, the formula corresponding to independent dominating sets is given by
\begin{verbatim}
all1 x: validVertex(x) => ( all2 y:validVertex(y) =>
  ((x in S & y in S) => ~edge(x,y)));
  
all1 x: validVertex(x) => 
  (x in S | ex1 y: validVertex(y) & edge(x,y) & y in S);
\end{verbatim}
while the equivalent $MSO$ formula was 
$$\left(\forall x,y: (x\in S\land y\in S)\implies \lnot E(x,y)\right)
\land\left(\forall x: x\notin S \implies \exists y: y\in S \land E(x,y)\right)$$

\section{The convex sets}\label{annexconvexsets}
We provide in this Annex the vertices of some the convex sets that we used in the different applications.
\subsection{Total perfect dominating sets}\label{tpdset}
Let  $\alpha = (2^{27}\times 7)^\frac{1}{85}\approx1.275157$, then the set of vertices used in Proposition \ref{treetpd} is

\begin{align*}
&\left\{\begin{pmatrix}0 \\ 0 \\ 1/939524096\alpha^{84} \\ 1/939524096\alpha^{84}\end{pmatrix},
\begin{pmatrix}0 \\ 1/939524096\alpha^{83} \\ 0 \\ 1/939524096\alpha^{83}\end{pmatrix},
\begin{pmatrix}0 \\ 1/469762048\alpha^{82} \\ 0 \\ 1/939524096\alpha^{82}\end{pmatrix},
\begin{pmatrix}0 \\ 3/939524096\alpha^{81} \\ 0 \\ 1/939524096\alpha^{81}\end{pmatrix},\right.\\&
\begin{pmatrix}0 \\ 1/234881024\alpha^{80} \\ 0 \\ 1/939524096\alpha^{80}\end{pmatrix},
\begin{pmatrix}1/939524096\alpha^{78} \\ 0 \\ 1/234881024\alpha^{78} \\ 0\end{pmatrix},
\begin{pmatrix}1/117440512\alpha^{72} \\ 0 \\ 1/58720256\alpha^{72} \\ 0\end{pmatrix},
\begin{pmatrix}3/58720256\alpha^{66} \\ 0 \\ 1/14680064\alpha^{66} \\ 0\end{pmatrix},\\&
\begin{pmatrix}1/3670016\alpha^{60} \\ 0 \\ 1/3670016\alpha^{60} \\ 0\end{pmatrix},
\begin{pmatrix}5/3670016\alpha^{54} \\ 0 \\ 1/917504\alpha^{54} \\ 0\end{pmatrix},
\begin{pmatrix}3/458752\alpha^{48} \\ 0 \\ 1/229376\alpha^{48} \\ 0\end{pmatrix},
\begin{pmatrix}1/32768\alpha^{42} \\ 0 \\ 1/57344\alpha^{42} \\ 0\end{pmatrix},
\begin{pmatrix}1/7168\alpha^{36} \\ 0 \\ 1/14336\alpha^{36} \\ 0\end{pmatrix},\\&
\begin{pmatrix}3/939524096\alpha^{80} \\ 1/939524096\alpha^{80} \\ 0 \\ 0\end{pmatrix},
\begin{pmatrix}9/14336\alpha^{30} \\ 0 \\ 1/3584\alpha^{30} \\ 0\end{pmatrix},
\begin{pmatrix}1/234881024\alpha^{79} \\ 1/939524096\alpha^{79} \\ 0 \\ 0\end{pmatrix},
\begin{pmatrix}5/1792\alpha^{24} \\ 0 \\ 1/896\alpha^{24} \\ 0\end{pmatrix},
\begin{pmatrix}11/896\alpha^{18} \\ 0 \\ 1/224\alpha^{18} \\ 0\end{pmatrix},\\&\left.
\begin{pmatrix}3/56\alpha^{12} \\ 0 \\ 1/56\alpha^{12} \\ 0\end{pmatrix},
\begin{pmatrix}13/56\alpha^{6} \\ 0 \\ 1/14\alpha^{6} \\ 0\end{pmatrix},
\begin{pmatrix}1 \\ 0 \\ 2/7 \\ 0\end{pmatrix}\right\}
\end{align*}

\subsection{3-matchings}\label{m3sets}
 Let $\alpha\approx1.3802$ be the real root of $x^4-x^3-1$ between 1 and 2, then the set of vertices used in Proposition \ref{3matchings} is
\begin{align*}
 &\left\{\begin{pmatrix}
\alpha^3 - 2\alpha^2 + \alpha \\
3\alpha^3 - 6\alpha^2 + 3\alpha \\
  \alpha^3 - 2\alpha^2 + \alpha \\
  \alpha^3 - 2\alpha^2 + \alpha 
\end{pmatrix},\begin{pmatrix}
 - \alpha^3 + \alpha^2 + 1\\
   0 \\
   - 3\alpha^3 + 3\alpha^2 + 3  \\
  - \alpha^3 + \alpha^2 + 1 
\end{pmatrix},\begin{pmatrix}   
 - \alpha^3 + \alpha^2 + 1  \\
  - 2\alpha^3 + 2\alpha^2 + 2  \\
  - \alpha^3 + \alpha^2 + 1  \\
  - \alpha^3 + \alpha^2 + 1 
\end{pmatrix},\begin{pmatrix}   
\alpha - 1 \\
  0 \\
  2\alpha - 2  \\
 \alpha - 1  
\end{pmatrix}\right.,\\
&\left.\begin{pmatrix}  
\alpha - 1 \\
  \alpha - 1  \\
 \alpha - 1 \\
  \alpha - 1  
\end{pmatrix},\begin{pmatrix}  
4\alpha^3 - 4\alpha^2 + 4\alpha - 8 \\
  4\alpha^3 - 4\alpha^2 + 4\alpha - 8 \\
  2\alpha^3 - 2\alpha^2 + 2\alpha - 4 \\
  \alpha^3 - \alpha^2 + \alpha - 2  
\end{pmatrix},\begin{pmatrix}  
\alpha^2 - \alpha \\
  0  \\
 \alpha^2 - \alpha \\
  \alpha^2 - \alpha    
\end{pmatrix},\begin{pmatrix}
 - 2\alpha^3 + 2\alpha^2 + 2  \\
  - \alpha^3 + \alpha^2 + 1  \\
  - \alpha^3 + \alpha^2 + 1  \\
  - \alpha^3 + \alpha^2 + 1    
\end{pmatrix},\begin{pmatrix}
\alpha^3 - \alpha^2  \\
 0  \\
 0  \\
 \alpha^3 - \alpha^2    
\end{pmatrix}\right\}
\end{align*}
\subsection{4-matchings}\label{m4sets}
Let $\alpha=13^\frac{1}{9}\approx 1.329754$, then the set of vertices used in Proposition \ref{4matchings} is
\begin{align*}
 &\left\{\begin{pmatrix}0\\ 1/13\alpha^8 + 1/6\\ 0\\ 0\\ 0\end{pmatrix},
\begin{pmatrix}1/13\alpha^3\\ 0\\ 4/13\alpha^3\\ 1/13\alpha^3\\ 1/13\alpha^3\end{pmatrix},
\begin{pmatrix}1/13\alpha^4\\ 0\\ 0\\ 4/13\alpha^4\\ 1/13\alpha^4\end{pmatrix},
\begin{pmatrix}1/13\alpha^4 \\0 \\3/13\alpha^4\\ 1/13\alpha^4\\ 1/13\alpha^4\end{pmatrix},
\begin{pmatrix}1/13\alpha^4\\ 2/13\alpha^4\\ 1/13\alpha^4\\ 1/13\alpha^4\\ 1/13\alpha^4\end{pmatrix},\right.\\&
\begin{pmatrix}4/13\\ 5/13 \\4/13\\ 2/13 \\1/13\end{pmatrix},
\begin{pmatrix}1/13\alpha^5 \\0 \\0 \\3/13\alpha^5 \\1/13\alpha^5\end{pmatrix},
\begin{pmatrix}1/13\alpha^5 \\0 \\2/13\alpha^5 \\1/13\alpha^5\\ 1/13\alpha^5\end{pmatrix},
\begin{pmatrix}1/13\alpha^5 \\1/13\alpha^5\\ 1/13\alpha^5\\ 1/13\alpha^5 1/13\alpha^5\end{pmatrix},
\begin{pmatrix}1/13\alpha^6 \\0 \\0 \\2/13\alpha^6\\ 1/13\alpha^6\end{pmatrix},\\&\left.
\begin{pmatrix}1/13\alpha^6 \\0 \\1/13\alpha^6 1/13\alpha^6\\ 1/13\alpha^6\end{pmatrix},
\begin{pmatrix}2/13\alpha^4 \\1/13\alpha^4\\ 1/13\alpha^4\\ 1/13\alpha^4 1/13\alpha^4\end{pmatrix},
\begin{pmatrix}1/13\alpha^7 \\0 \\0 \\1/13\alpha^7\\ 1/13\alpha^7\end{pmatrix},
\begin{pmatrix}1/13\alpha^8 \\0 \\0 \\0 \\1/13\alpha^8 \end{pmatrix}\right\}
\end{align*}

\subsection{Maximal matchings}\label{maxmatch}
Let  $\alpha =\left(\left(\frac{11+\sqrt85}{2}\right)^\frac{n}{7}\right)\approx  1.391664$, then the set of vertices used in Proposition \ref{simplematchings} is
\begin{align*}
&\left\{\begin{pmatrix}0 \\   - 1/9\alpha^{13} + 11/9\alpha^6 \\  0 \\   - 1/9\alpha^{13} + 11/9\alpha^6 \end{pmatrix},
\begin{pmatrix} - 4/765\alpha^{10} + 107/765\alpha^3 \\  5/153\alpha^{10} - 19/153\alpha^3 \\  0 \\  5/153\alpha^{10} - 19/153\alpha^3 \end{pmatrix},
\begin{pmatrix} - 1/9\alpha^{10} + 11/9\alpha^3 \\   - 2/9\alpha^{10} + 22/9\alpha^3 \\  0 \\   - 2/9\alpha^{10} + 22/9\alpha^3 \end{pmatrix}\right.,\\&
\begin{pmatrix} - 2/9\alpha^9 + 22/9\alpha^2 \\   - 1/9\alpha^9 + 11/9\alpha^2 \\   - 2/9\alpha^9 + 22/9\alpha^2 \\   - 1/3\alpha^9 + 11/3\alpha^2 \end{pmatrix},
\begin{pmatrix}5/153\alpha^9 - 19/153\alpha^2 \\   - 4/765\alpha^9 + 107/765\alpha^2 \\  5/153\alpha^9 - 19/153\alpha^2 \\  7/255\alpha^9 + 4/255\alpha^2 \end{pmatrix},
\begin{pmatrix} - 1/9\alpha^{12} + 11/9\alpha^5 \\  0 \\   - 1/9\alpha^{12} + 11/9\alpha^5 \\   - 1/9\alpha^{12} + 11/9\alpha^5 \end{pmatrix},\\&
\begin{pmatrix} - 44/27\alpha^{12} + 448/27\alpha^5 \\   - 11/81\alpha^{12} + 112/81\alpha^5 \\   - 88/81\alpha^{12} + 896/81\alpha^5 \\   - 11/9\alpha^{12} + 112/9\alpha^5 \end{pmatrix},
\begin{pmatrix}29/765\alpha^{12} - 202/765\alpha^5 \\   - 107/6885\alpha^{12} + 1141/6885\alpha^5 \\  71/6885\alpha^{12} - 178/6885\alpha^5 \\   - 4/765\alpha^{12} + 107/765\alpha^5 \end{pmatrix},
\begin{pmatrix}32/2295\alpha^{12} - 46/2295\alpha^5 \\   - 8/6885\alpha^{12} + 133/6885\alpha^5 \\  71/6885\alpha^{12} - 178/6885\alpha^5 \\  7/765\alpha^{12} - 1/153\alpha^5 \end{pmatrix},\\&
\begin{pmatrix} - 5/9\alpha^8 + 55/9\alpha \\  0 \\   - 1/3\alpha^8 + 11/3\alpha \\   - 1/3\alpha^8 + 11/3\alpha \end{pmatrix},
\begin{pmatrix}46/765\alpha^8 - 83/765\alpha \\  0 \\  7/255\alpha^8 + 4/255\alpha \\  7/255\alpha^8 + 4/255\alpha \end{pmatrix},
\begin{pmatrix} - 2/9\alpha^{11} + 22/9\alpha^4 \\  0 \\   - 1/9\alpha^{11} + 11/9\alpha^4 \\   - 1/9\alpha^{11} + 11/9\alpha^4 \end{pmatrix},\\&
\begin{pmatrix}253/4335\alpha^8 - 227/4335\alpha \\  283/585225\alpha^8 + 1036/585225\alpha \\  3178/117045\alpha^8 - 616/23409\alpha \\  599/21675\alpha^8 - 532/21675\alpha \end{pmatrix},
\begin{pmatrix}5/153\alpha^{11} - 19/153\alpha^4 \\  0 \\   - 4/765\alpha^{11} + 107/765\alpha^4 \\   - 4/765\alpha^{11} + 107/765\alpha^4 \end{pmatrix},
\begin{pmatrix}53/2295\alpha^{11} - 61/2295\alpha^4 \\  0 \\  7/765\alpha^{11} - 1/153\alpha^4 \\  7/765\alpha^{11} - 1/153\alpha^4 \end{pmatrix},\\&\left.
\begin{pmatrix} - 8/9\alpha^7 + 88/9 \\  0 \\   - 1/3\alpha^7 + 11/3 \\   - 1/3\alpha^7 + 11/3 \end{pmatrix},
\begin{pmatrix}67/765\alpha^7 - 71/765 \\  0 \\  7/255\alpha^7 + 4/255 \\  7/255\alpha^7 + 4/255 \end{pmatrix},
\begin{pmatrix} - 1/3\alpha^{10} + 11/3\alpha^3 \\  0 \\   - 1/9\alpha^{10} + 11/9\alpha^3 \\   - 1/9\alpha^{10} + 11/9\alpha^3 \end{pmatrix}\right\}
\end{align*}

\section{The bilinear system of maximal irredundant sets}\label{annex_irr_sets}
The initial vector, final vector and the bilinear map corresponding to maximal irredundant sets that are used in Proposition \ref{Max_irr_sets} are respectively given by
$$
\mathbf{V_0}=\begin{pmatrix}
0 \\0 \\0 \\0 \\0 \\1 \\0 \\0 \\0 \\0 \\0 \\0 \\0 \\0 \\1 \\0 \\0 \\0 \\0 \\0 
\end{pmatrix}, \mathbf{F}=
\begin{pmatrix}
  1 \\1 \\1 \\1 \\1 \\1 \\1 \\1 \\0 \\0 \\0 \\0 \\0 \\0 \\0 \\0 \\0 \\0 \\0 \\0
\end{pmatrix}\text{ and for all } u,v\in\mathbb{R}^{20},
$$

$$
\resizebox{\hsize}{!}{$
\mathbf{B}(x,y)=\begin{pmatrix}
 u_{1}v_{1}+ u_{1}v_{2}+ u_{1}v_{3}+ u_{1}v_{4}+ u_{1}v_{5}+ u_{1}v_{6}+ u_{1}v_{7}+ u_{1}v_{8}+ u_{2}v_{5}+ u_{2}v_{6}+ u_{2}v_{8}+ u_{7}v_{3}+ u_{7}v_{6}+ u_{9}v_{6}+ u_{9}v_{8}+ u_{10}v_{3}+ u_{10}v_{4}+ u_{10}v_{6}+ u_{10}v_{8}+ u_{11}v_{5}+ u_{11}v_{6}+ u_{11}v_{8}+ u_{12}v_{6}+ u_{12}v_{8}+ u_{15}v_{6}+ u_{16}v_{6}+ u_{17}v_{6}+ u_{18}v_{3}+ u_{18}v_{6}+ u_{19}v_{6}\\
  u_{2}v_{1}+ u_{2}v_{2}+ u_{2}v_{3}+ u_{2}v_{4}+ u_{9}v_{3}+ u_{9}v_{4}+ u_{15}v_{3}+ u_{17}v_{3}\\
  u_{3}v_{1}+ u_{3}v_{2}+ u_{3}v_{3}+ u_{3}v_{4}+ u_{3}v_{7}+ u_{3}v_{9}+ u_{3}v_{10}+ u_{3}v_{13}+ u_{3}v_{15}+ u_{3}v_{17}+ u_{3}v_{18}+ u_{6}v_{7}+ u_{6}v_{15}+ u_{6}v_{17}+ u_{6}v_{18}+ u_{14}v_{7}+ u_{14}v_{15}+ u_{14}v_{17}+ u_{14}v_{18}\\
  u_{4}v_{1}+ u_{4}v_{2}+ u_{4}v_{3}+ u_{4}v_{4}+ u_{4}v_{9}+ u_{4}v_{10}+ u_{4}v_{13}+ u_{6}v_{16}+ u_{6}v_{19}+ u_{13}v_{3}+ u_{13}v_{4}+ u_{13}v_{13}+ u_{14}v_{16}+ u_{14}v_{19}+ u_{14}v_{20}\\
  u_{5}v_{1}+ u_{5}v_{2}+ u_{5}v_{5}+ u_{5}v_{7}+ u_{5}v_{11}+ u_{5}v_{15}+ u_{5}v_{16}+ u_{7}v_{14}+ u_{15}v_{14}+ u_{16}v_{14}+ u_{17}v_{14}+ u_{18}v_{14}+ u_{19}v_{14}+ u_{20}v_{14}\\
  u_{6}v_{1}+ u_{6}v_{2}+ u_{6}v_{9}+ u_{6}v_{10}\\
  u_{7}v_{1}+ u_{7}v_{5}+ u_{7}v_{7}+ u_{15}v_{5}+ u_{16}v_{5}\\
  u_{6}v_{11}+ u_{6}v_{12}+ u_{8}v_{1}+ u_{8}v_{2}+ u_{8}v_{9}+ u_{8}v_{10}+ u_{8}v_{11}+ u_{8}v_{12}\\
  u_{9}v_{1}+ u_{9}v_{2}+ u_{15}v_{4}+ u_{17}v_{4}\\
  u_{7}v_{4}+ u_{7}v_{8}+ u_{9}v_{5}+ u_{10}v_{1}+ u_{10}v_{2}+ u_{10}v_{5}+ u_{10}v_{7}+ u_{12}v_{5}+ u_{15}v_{8}+ u_{16}v_{8}+ u_{17}v_{8}+ u_{18}v_{4}+ u_{18}v_{8}+ u_{19}v_{8}\\
  u_{2}v_{7}+ u_{11}v_{1}+ u_{11}v_{2}+ u_{11}v_{3}+ u_{11}v_{4}+ u_{11}v_{7}+ u_{12}v_{3}+ u_{12}v_{4}+ u_{16}v_{3}+ u_{19}v_{3}\\
  u_{9}v_{7}+ u_{12}v_{1}+ u_{12}v_{2}+ u_{12}v_{7}+ u_{16}v_{4}+ u_{19}v_{4}\\
  u_{6}v_{20}+ u_{13}v_{1}+ u_{13}v_{2}+ u_{13}v_{9}+ u_{13}v_{10}\\
  u_{6}v_{3}+ u_{6}v_{4}+ u_{6}v_{13}+ u_{14}v_{1}+ u_{14}v_{2}+ u_{14}v_{3}+ u_{14}v_{4}+ u_{14}v_{9}+ u_{14}v_{10}+ u_{14}v_{13}\\
  u_{15}v_{1}\\
  u_{15}v_{7}+ u_{16}v_{1}+ u_{16}v_{7}\\
  u_{15}v_{2}+ u_{17}v_{1}+ u_{17}v_{2}\\
  u_{7}v_{2}+ u_{17}v_{5}+ u_{18}v_{1}+ u_{18}v_{2}+ u_{18}v_{5}+ u_{18}v_{7}+ u_{19}v_{5}\\
  u_{16}v_{2}+ u_{17}v_{7}+ u_{19}v_{1}+ u_{19}v_{2}+ u_{19}v_{7}\\
  u_{7}v_{11}+ u_{7}v_{15}+ u_{7}v_{16}+ u_{15}v_{11}+ u_{15}v_{15}+ u_{15}v_{16}+ u_{16}v_{11}+ u_{16}v_{15}+ u_{16}v_{16}+ u_{17}v_{11}+ u_{17}v_{15}+ u_{17}v_{16}+ u_{18}v_{11}+ u_{18}v_{15}+ u_{18}v_{16}+ u_{19}v_{11}+ u_{19}v_{15}+ u_{19}v_{16}+ u_{20}v_{1}+ u_{20}v_{2}+ u_{20}v_{5}+ u_{20}v_{7}+ u_{20}v_{11}+ u_{20}v_{15}+ u_{20}v_{16}\\
\end{pmatrix}
 $ }
$$

\end{document}